\documentclass[a4paper,UKenglish,cleveref,autoref,thm-restate,biblatex,pdfa]{lipics-v2021}

\pdfoutput=1 
\hideLIPIcs  

\title{Commutative algebras of series}

\author
    {Lorenzo Clemente}
    {University of Warsaw, Department of Mathematics, Mechanics, and Computer Science}
    {clementelorenzo@gmail.com}
    {https://orcid.org/0000-0003-0578-9103}
    {Partially supported by the Polish National Science Centre grant no.~2024/54/E/ST6/00287}

\authorrunning{Lorenzo Clemente} 

\Copyright{Lorenzo Clemente} 



\category{} 

\acknowledgements{I would like to thank Antonio Casares, Bartek Klin, Neha Rino, and Szymon Toruńczyk
for their comments on different aspects of this work.
I would also like to thank the anonymous reviewers for their constructive comments.}

\nolinenumbers 

\usepackage[utf8]{inputenc}
\usepackage{stmaryrd}
\usepackage{amsthm, amsmath}
\usepackage{mathtools}
\usepackage{thm-restate}
\usepackage{xspace}
\usepackage[textwidth=2cm,textsize=small]{todonotes}
\usepackage{hyperref}
\usepackage{stackrel}
\usepackage{MnSymbol}
\usepackage{fix-cm}
\usepackage{shuffle}
\usepackage{paralist}
\usepackage[inline]{enumitem}
\usepackage{graphicx}
\usepackage{etoolbox}
\usepackage{enumitem}
\usepackage[bb=boondox]{mathalfa}
\usepackage{multirow}
\usepackage{amstext} 
\usepackage{wrapfig}
\usepackage{makecell}
\setcellgapes{2pt}
\usepackage{textcase} 

\DeclareFontFamily{U}{shuffle}{}
\DeclareFontShape{U}{shuffle}{m}{n}{ <-8>shuffle7 <8->shuffle10}{}

\usepackage{array}   
\newcolumntype{C}{>{$}c<{$}} 
\newcolumntype{L}{>{$}l<{$}} 

\hypersetup{
    colorlinks,
    linkcolor={red!50!black},
    citecolor={blue!50!black},
    urlcolor={blue!80!black}
}

\usepackage{tikz}
\usetikzlibrary{positioning,shapes,calc,arrows.meta}

\makeatletter
\patchcmd{\@sect}{\uppercase}{\MakeTextUppercase}{}{}
\patchcmd{\@sect}{\uppercase}{\MakeTextUppercase}{}{}
\makeatother

\newcommand{\baseurl}{https://lclem.github.io/lics2026-agda}
\newcommand{\agda}[2]{\protect{\href{\NoCaseChange{\baseurl/#2}}{#1}}}

\newcommand{\agdasection}[2]{
  \section{\agda{#1}{#2}}
}

\newcommand{\agdasubsubsection}[2]{
  \subsubsection{\agda{#1}{#2}}
}

\newcommand{\tuple}[1]{\left(#1\right)}
\newcommand{\tuplesmall}[1]{(#1)}

\newcommand{\e}{\varepsilon}
\newcommand{\sep}{\;\mid\;}

\newcommand{\ignore}[1]{}

\newcommand{\wrt}{w.r.t.}
\newcommand{\cf}{cf.}
\newcommand{\eg}{e.g.}
\newcommand{\aka}{a.k.a.}

\newif\ifstartedinmathmode
\renewcommand*{\st}{
  \relax\ifmmode\startedinmathmodetrue\else\startedinmathmodefalse\fi
  \ifstartedinmathmode{\;\cdot\;}\else{s.t.}\fi%
}

\newcommand{\N}{\mathbb N}

\newcommand{\Q}{\mathbb Q}

\renewcommand{\AA}{\mathcal A}
\newcommand{\BB}{\mathcal B}

\renewcommand{\SS}{\mathcal S}

\newcommand{\Terms}[1]{\mathsf{Terms}\,{#1}}
\newcommand{\Lin}[1]{\mathsf{Lin}\,{#1}}

\let\oldcirc\circ
\let\circ\relax
\DeclareMathOperator{\circ}{\oldcirc\,}

\DeclareMathOperator{\hadamard}{\odot}
\DeclareMathOperator{\infiltration}{\uparrow}

\newcommand{\card}[1]{\left|#1\right|}
\newcommand{\length}[1]{\left|#1\right|}

\newcommand{\set}[1]{\left\{#1\right\}}
\newcommand{\setof}[2]{\set{#1 \;\middle|\; #2}}

\newcommand{\poly}[2]{#1[#2]}
\newcommand{\hompoly}[2]{\poly {#1} {#2}_0}

\newcommand{\series}[2]{#1\langle\!\langle#2\rangle\!\rangle}

\newcommand{\ideal}[1]{\langle#1\rangle}
\newcommand{\idealof}[2]{\ideal{#1 \mid #2}}

\newcommand{\sem}[1]{\llbracket#1\rrbracket}

\newcommand{\Asem}[2]{{#1\sem{#2}}}

\newcommand{\coefficient}[2]{[#1]#2}

\newcommand{\algebra}[2]{#1\{#2\}}
\newcommand{\Val}[1]{\mathsf{Val}\,{#1}}

\newcommand{\leftright}[4]{\mbox{}^{#1}{#2}_{#3}^{#4}}
\newcommand{\lr}[3]{\leftright{#1}{#2}{}{#3}}
\newcommand{\lxr}[4]{\leftright{#1}{#2}{#3}{#4}}

\newcommand{\one}{\mathbb 1}
\newcommand{\zero}{\mathbb 0}

\newcommand{\derive}[1]{\delta_{#1}}
\newcommand{\deriveleft}[1]{\derive{#1}}
\newcommand{\deriveright}[1]{\delta^R_{#1}}
\newcommand{\reverse}[1]{{#1}^\mathsf{rev}}

\let\oldpartial\partial
\renewcommand{\partial}[1]{\oldpartial_{#1}}

\makeatletter
\def\smallunderbrace#1{\mathop{\vtop{\m@th\ialign{##\crcr
   $\hfil\displaystyle{#1}\hfil$\crcr
   \noalign{\kern3\p@\nointerlineskip}%
   \tiny\upbracefill\crcr\noalign{\kern3\p@}}}}\limits}
\makeatother

\newcommand{\BAC}{\textsf{BAC}}

\newenvironment{proofsketch}{\begin{proof}[Proof sketch]}{\end{proof}}
\newenvironment{challenge}{\begin{quote}\centering\scshape}{\end{quote}}

\AtEndPreamble{%

  \theoremstyle{acmdefinition}

  \declaretheoremstyle[
    spaceabove=6pt,
    spacebelow=6pt,
    headfont=\normalfont\bfseries,
    notefont=\normalfont\itshape,
    bodyfont=\normalfont,
    postheadspace=0.5em,
    headformat={\protect{\agda{\textsc{Theorem~\NUMBER}}{\agdaurl}}\NOTE}
  ]{agdatheoremstyle}

  \declaretheoremstyle[
    spaceabove=6pt,
    spacebelow=6pt,
    headfont=\normalfont\bfseries,
    notefont=\normalfont\itshape,
    bodyfont=\normalfont,
    postheadspace=0.5em,
    headformat={\protect{\agda{\textsc{Lemma~\NUMBER}}{\agdaurl}}\NOTE}
  ]{agdalemmastyle}

  \declaretheoremstyle[
    spaceabove=6pt,
    spacebelow=6pt,
    headfont=\normalfont\bfseries,
    notefont=\normalfont\itshape,
    bodyfont=\normalfont,
    postheadspace=0.5em,
    headformat={\protect{\agda{\textsc{Corollary~\NUMBER}}{\agdaurl}}\NOTE}
  ]{agdacorollarystyle}

  \declaretheoremstyle[
    spaceabove=6pt,
    spacebelow=6pt,
    headfont=\normalfont\bfseries,
    notefont=\normalfont\itshape,
    bodyfont=\normalfont,
    postheadspace=0.5em,
    headformat={\protect{\agda{\textsc{Proposition~\NUMBER}}{\agdaurl}}\NOTE}
  ]{agdapropositionstyle}

  \declaretheorem[
    style=agdalemmastyle,
    name=Lemma,
    sibling=theorem
  ]{agdalem}

  \declaretheorem[
    style=agdacorollarystyle,
    name=Corollary,
    sibling=theorem
  ]{agdacor}

  \declaretheorem[
    style=agdapropositionstyle,
    name=Proposition,
    sibling=theorem
  ]{agdaprop}

  \declaretheorem[
    style=agdatheoremstyle,
    name=Theorem,
    sibling=theorem
  ]{agdathm}

  \newcommand{\agdaurl}{}

  \newenvironment{agdalemma}[1][]{
    \renewcommand{\agdaurl}{#1}
    \begin{agdalem}}
      {\end{agdalem}}

  \newenvironment{agdacorollary}[1][]{
    \renewcommand{\agdaurl}{#1}
    \begin{agdacor}}
    {\end{agdacor}}

  \newenvironment{agdaproposition}[1][]{%
    \renewcommand{\agdaurl}{#1}%
    \begin{agdaprop}}
    {\end{agdaprop}}

}

\crefname{equation}{}{}
\crefname{definition}{Definition}{Definitions}
\crefname{claim}{Claim}{Claims}
\crefname{section}{§}{§§}
\crefname{proposition}{Proposition}{Propositions}
\crefname{agdaprop}{Proposition}{Propositions}
\crefname{lemma}{Lemma}{Lemmas}
\crefname{agdalem}{Lemma}{Lemmas}
\crefname{corollary}{Corollary}{Corollaries}
\crefname{agdacor}{Corollary}{Corollaries}
\crefname{theorem}{Theorem}{Theorems}
\crefname{agdathm}{Theorem}{Theorems}
\crefname{figure}{Figure}{Figures}
\crefname{table}{Table}{Tables}
\crefname{fact}{Fact}{Facts}
\crefname{example}{Example}{Examples}
\crefname{remark}{Remark}{Remarks}
\crefname{item}{}{}
\crefname{enumi}{}{}

\begin{document}

\maketitle

\begin{abstract}
  We consider a large family of product operations of formal power series in noncommuting indeterminates,
  the classes of automata they define, and the respective equivalence problems.
  A \emph{$P$-product} of series is defined coinductively by a \emph{polynomial product rule $P$},
  which gives a recursive recipe to build the product of two series
  as a function of the series themselves and their derivatives.
  %
  %
  %
  %
  %
  %
  The first main result of the paper is a complete and decidable characterisation of all product rules $P$
  giving rise to $P$-products which are bilinear, associative, and commutative.
  The characterisation shows that there are infinitely many such products,
  and in particular it applies to the notable Hadamard, shuffle, and infiltration products from the literature.
  
  Every $P$-product gives rise to the class of \emph{$P$-automata},
  an infinite-state model where states are terms.
  The second main result of the paper is that the equivalence problem for $P$-automata is decidable
  for $P$-products satisfying our characterisation.
  This explains, subsumes, and extends previous results
  on Hadamard, shuffle, and infiltration automata.
  We have formalised most results in the proof assistant Agda.
\end{abstract}


\section{Introduction}

\subparagraph*{Rational series.}

Let $\Sigma = \set{a_1, \dots, a_d}$ be a finite alphabet of non-commuting indeterminates.
A \emph{series} over the rationals is a mapping $\Sigma^* \to \Q$
from the set of finite words over $\Sigma$ to the rationals.
Sometimes series are called \emph{weighted languages} and are denoted by $\series \Q \Sigma$~\cite{BerstelReutenauer:CUP:2010}.
The set of series has a \emph{dynamic structure},
given by the family of \emph{left derivatives} $\deriveleft a : \series \Q \Sigma \to \series \Q \Sigma$ ($a \in \Sigma$),
and it can be endowed with a variety of \emph{algebraic structures}.
For instance, series form a $\Q$-vector space (formal definitions will be given later),
\begin{align}
    \label{eq:vector space structure}
    \tag{vector space}
    \tuplesmall{\series \Q \Sigma;
        \smallunderbrace {\vphantom{\Large|}{\zero}}_{{\text{zero}}},\ 
        \smallunderbrace {\vphantom{\Large|}{(c \cdot {\_})_{c \in \Q}}}_{\text{scalar multiplication}},\ 
        \smallunderbrace {\vphantom{\Large|}{+}}_{\text{addition}}}.
\end{align}
Each choice for the algebraic structure
gives rise to a corresponding \emph{algebraically-defined class of series}.
For instance, the vector space structure above
gives rise to the subclass of \emph{rational series}~\cite{Schutzenberger:IC:1961},
by considering series that belong to finite-dimensional linear subspaces closed under left derivatives.
Rational series appear in various areas of science and mathematics~\cite{HandbookWA},
with important applications in control theory~\cite{Fliess:MST:1976,Fliess:1981}.
%
%
Most decision problems for rational series are undecidable
\cite[Theorem 21]{NasuHonda:IC:1969}~(\cf~also~\cite[Theorem 6.17]{Paz:1971}),
however Schützenberger remarkably showed that \emph{equivalence} is decidable%
\footnote{%
This follows as a special case of a more general \emph{effective minimisation} result~\cite[Sec.~III.B]{Schutzenberger:IC:1961}.
}.
The latter result is obtained by considering \emph{weighted finite automata},
a syntactic model recognising precisely the rational series,
which can be used to devise an equivalence algorithm.

Summarising, an important part of the theory of series can be developed along the following four main axes:
\begin{inparaenum}
    \item[(A)] algebraic structure, 
    \item[(B)] algebraically-defined classes of series, 
    \item[(C)] syntactically-defined classes of series, and 
    \item[(D)] the equivalence problem. 
\end{inparaenum}
In the case of rational series, this is illustrated in~\cref{fig:rational series}.
The success of rational series has motivated the study of more general classes retaining decidability of equivalence.
One principled way to proceed consists in introducing additional algebraic structure, as we explain next.

\begin{figure}[!h]
    \begin{center}
        \tikzstyle{box} = [bgLine, fill=white, inner sep=0.8mm, anchor=north west, outer sep=0, text=black]
        \begin{tikzpicture}[
            node distance=2ex,
            block/.style ={rectangle, draw=black, thick, fill=blue!20, text width=10.5em, align=center, rounded corners, minimum height=3em, minimum width=10ex},
            every edge/.style = {draw, thick} 
            ]
        \node[block] (A) {(A) algebraic structure \\ \em $\Q$-vector space};
        \node[block] (B) [right=2em of A] {(B) algebraic class \\ \em rational series};
        \node[block] (C) [below=2em of A, xshift=14ex] {(C) syntactic class \\ \em weighted automata};
        \node[block] (D) [below=2em of C, text width=12em] {(D) equivalence problem \\ \em linear-algebraic methods};

        \path (A) edge[->] (B);
        \path (A) edge[->] (C);
        \path (B) edge[<->] (C); 
        \path (C) edge[->] (D);
        \end{tikzpicture}
    \end{center}
    \caption{Theory of rational series}
    \label{fig:rational series}
\end{figure}

\subparagraph*{(A) Algebras of series.}
We extend the vector space structure on series by considering (binary) product operations ``$*$''
which are bilinear, associative, and commutative (\emph{\BAC}).
Such a product endows the set of series with the structure of a \emph{commutative $\Q$-algebra}%
\footnote{We do not require the existence of a multiplicative identity,
thus series algebras need not be unital.}
\begin{align}
    \label{eq:series algebra}
    \tag{series $\Q$-algebra}
    \tuplesmall{\series \Q \Sigma;
        \smallunderbrace {\vphantom{\Large|}\zero}_{\text{zero}},\ 
        \smallunderbrace {\vphantom{\Large|}{(c \cdot {\_})_{c \in \Q}}}_{\text{scalar multiplication}},\ 
        \smallunderbrace {\vphantom{\Large|}{+}}_{\text{addition}},\ 
        \smallunderbrace {\vphantom{\Large|}{*}}_{\text{product}}}.
\end{align}
Notable examples 
are the \emph{Hadamard}, \emph{shuffle}, and \emph{infiltration} products~\cite{Fliess:1974,ChenFoxLyndon:AM:1958}.
%
Remarkably, they all satisfy the \emph{constant term rule} $(f * g)(\e) = f(\e) \cdot g(\e)$
and a \emph{product rule} of the form
\begin{align}
    \tag{product rule}
    \label{eq:product rule}
    \deriveleft a (f * g)
        &= P(f, \deriveleft a f, g, \deriveleft a g),
\end{align}
where $P$ is an expression built from the mentioned series, scalar multiplication, addition,
and product~\cite[Sec.~3]{BasoldHansenPinRutten:MSCS:2017}.
In fact, the two rules above uniquely define the product,
which we call a \emph{$P$-product}.
For instance, shuffle ``$\shuffle$'' obeys the Leibniz rule from calculus
$\deriveleft a (f \shuffle g) = (\deriveleft a f) \shuffle g + f \shuffle (\deriveleft a g)$,
and thus it is a $P$-product for $P(x, \dot x, y, \dot y) = \dot x y + x \dot y$.
Not all $P$-products are~\BAC,
for instance $\deriveleft a (f * g) = \deriveleft a f * g$ does not define a commutative operation.
This naturally leads to our first challenge.
\begin{challenge}
    \textbf{First challenge} \\
    Characterise the product rules $P$ yielding \BAC~$P$-products.
\end{challenge}


Restricting our attention to \BAC~$P$-products fixes the first axis (A),
and it is natural to extend the study of series along the remaining axes (B), (C), and (D); \cf~\cref{fig:P-finite series}.
This program will generate further challenges, as we explain next.

\begin{figure}[h!]
    \begin{center}
        \begin{tikzpicture}[
            node distance=2ex,
            block/.style ={rectangle, draw=black, thick, fill=blue!20, text width=11em, align=center, rounded corners, minimum height=3em},
            every edge/.style = {draw, thick}
            ]
        \node[block] (A) {(A) algebraic structure \\ \em commutative $\Q$-algebra};
        \node[block] (B) [right=2em of A] {(B) algebraic class \\ \em $P$-finite series};
        \node[block] (C) [below=2em of A, xshift=14ex] {(C) syntactic class \\ \em $P$-automata};
        \node[block, text width=13em] (D) [below=2em of C, ] {(D) equivalence problem \\ \em nonlinear-algebraic methods};

        \path (A) edge[->] (B);
        \path (A) edge[->] (C);
        \path (B) edge[<->] node[right, xshift=1ex] {} (C);
        \path (C) edge[->] (D);
        \end{tikzpicture}
    \end{center}
    \caption{Theory of $P$-finite series}
    \label{fig:P-finite series}
\end{figure}

\subparagraph*{(B) Algebraic classes of series: $P$-finite series.}

Every $P$-product (not necessarily \BAC) gives naturally rise
to an algebraically-defined class of series generalising the rational ones,
which we call \emph{$P$-finite series}%
\footnote{$P$-finite series, where $P$ is a product rule,
should not be confused with \emph{polynomially-finite sequences},
which sometimes are abbreviated as ``P-finite sequences''~\cite{Stanley:EJC:1980},
where ``P'' stands for ``polynomially''.}.
A series is $P$-finite if it belongs to a finite-dimensional subalgebra
closed under left derivatives.
This notion is inherently semantic,
and it is designed to ensure closure properties with respect to the algebra operations.
For instance, the notable products give rise to the classes of
\emph{Hadamard}, \emph{shuffle}, resp., \emph{infiltration-finite} series~\cite{Clemente:CONCUR:2024,Clemente:LICS:2025}.
These classes can also be motivated by counting problems in enumerative combinatorics,
for instance the Hadamard-finite series $a^n \mapsto 2^{2^n}$ counts the number of Boolean functions on $n$ variables
and the shuffle-finite series $a^n \mapsto n!$ counts the number of permutations on $n$ elements.
The drawback of $P$-finiteness is that it is not clear how it leads to algorithms for the equivalence problem.
For this reason, we seek a syntactic description of $P$-finite series.

\subparagraph*{(C) Syntactic classes of series: $P$-automata.}

Every $P$-product (not necessarily \BAC) naturally gives rise to a syntactic class of \emph{$P$-automata}, where
%
states are \emph{terms} built from variables and constructors for the algebra operations, and
transitions are defined via~\cref{eq:product rule}.
%
%
%
Generalising previous observations for the Hadamard, shuffle, and infiltration products~\cite{Clemente:CONCUR:2024,Clemente:LICS:2025},
we note that the class of series recognised by $P$-automata coincides with the $P$-finite series.

$P$-automata in general do not lead to an equivalence algorithm,
and in fact equivalence may well be undecidable for arbitrary product rules $P$.
In the case of a \BAC~$P$-product, the semantic domain is a commutative $\Q$-algebra of series,
and thus $P$-automata make too many syntactic distinctions in the state which are not reflected in the semantics.
For instance, $x y$ and $y x$ are syntactically different states,
however they recognise the same series by commutativity of the series product.
Therefore, it is natural to \emph{quotient} the state space of $P$-automata \wrt~the axioms of commutative $\Q$-algebras.
%
%
But there is a problem: Quotienting must respect the transition structure
and a priori it is not clear that this is the case.
If we could show that quotienting respects transitions,
then we would obtain \emph{polynomial} $P$-automata,
whose semantics is a homomorphism of commutative $\Q$-algebras.


In the case of the notable \BAC~$P$-products,
quotienting \emph{does} respect transitions,
leading to well-behaved polynomial $P$-automata models such as
\emph{Hadamard automata}~(called \emph{polynomial automata} in~\cite{BenediktDuffSharadWorrell:LICS:2017}),
\emph{shuffle automata}~\cite{Clemente:CONCUR:2024}, and \emph{infiltration automata}~\cite{Clemente:LICS:2025}.
This is a remarkable fact, which demands an explanation.
In the case of a single-letter alphabet,
Boreale and Gorla assume a technical condition
called \emph{well-behavedness}~\cite[Def.~3.5]{BorealeGorla:CONCUR:2021}, and show that
\begin{inparaenum}[(a)]
    \item it is satisfied by the notable products, and
    \item it implies the quotienting property~\cite[Lemma 3.6 and Theorem 3.7]{BorealeGorla:CONCUR:2021}.
\end{inparaenum}
%
%
However it is not clear whether this condition is \emph{necessary} for \BAC~products,
neither whether it is \emph{decidable}.
Our next challenge is to show the quotienting property \emph{for every \BAC~$P$-product},
without referring to any other condition (such as well-behavedness).

\begin{challenge}
    \textbf{Second challenge} \\
    Show that every \BAC~$P$-product
    gives rise to polynomial $P$-automata.
\end{challenge}


\subparagraph*{(D) Algorithms.}

It remains to see whether the equivalence problem for $P$-automata is decidable for \BAC~$P$-products.
This is the case for Hadamard~\cite[Corollary 1]{BenediktDuffSharadWorrell:LICS:2017},
shuffle~\cite[Theorem 1]{Clemente:CONCUR:2024}, and infiltration automata~\cite[Theorem 2]{Clemente:LICS:2025}.
All these works rely on \emph{Hilbert's finite basis theorem}, whose algorithmic application
appears already in the late 1990's work of Novikov and Yakovenko~\cite{NovikovYakovenko:1999}.
For a one-letter input alphabet $\card \Sigma = 1$,
Boreale and Gorla have used it to show decidability for well-behaved $P$-products~\cite[Theorem 4.1]{BorealeGorla:CONCUR:2021}.
However, this relied on an additional technical condition, which we may call \emph{ideal compatibility},
which is satisfied by the notable products.
%
Despite the fact that to date there is a single algorithmic approach to decide equivalence,
it is not clear whether it would work for arbitrary \BAC~$P$-products,
and whether ideal compatibility is necessary.



\begin{challenge}
    \textbf{Third challenge} \\
    Show that 
    equivalence is decidable for polynomial $P$-automata.
\end{challenge}

\subsection{Main results}
\label{sec:main results}


We are now ready to state our main results, addressing the challenges that we have proposed.

\subparagraph*{First challenge: Characterisation.}

Our first main contribution is a complete answer to our first challenge,
by means of an equational characterisation of \BAC~$P$-products.
\begin{theorem}[Characterisation]
    \label{thm:characterisation}
    %
    A $P$-product is bilinear, associative, and commutative
    if, and only if, the product rule $P$ satisfies the following identities%
    \footnote{We understand these identities as equality of terms up to bilinearity, associativity, and commutativity.
    In other words, they are equalities of the corresponding polynomials in the free commutative $\Q$-algebra generated by $x, \dot x, y, \dot y, z, \dot z$.}:
    \begin{align}
        \tag{$P\text{-add}$}
        \label{eq:left additivity}
        P(x + y, \dot x + \dot y, z, \dot z)
            &= P(x, \dot x, z, \dot z) + P(y, \dot y, z, \dot z) \\
        \tag{$P\text{-assoc}$}
        \label{eq:associativity}
        P(x, \dot x, yz, P(y, \dot y, z, \dot z))
            &= P(xy, P(x, \dot x, y, \dot y), z, \dot z) \\
        \tag{$P\text{-comm}$}
        \label{eq:commutativity}
        P(x, \dot x, y, \dot y)
            &= P(y, \dot y, x, \dot x).
    \end{align}
    %
\end{theorem}
The fact that this set of equations is sufficient is easy to prove
once we have a candidate characterisation (by coalgebraic methods),
however we remark that to prove each of the three properties
we need to use all equations together.
This shows that the interplay between additivity, associativity, and commutativity
is crucial for our proof to go through.

On the other hand, necessity can be proved for each equation separately,
using the fact that the operation mapping a series $f$ to its constant term $f(\e)$ is a homomorphism
and using the fact that the rational numbers $\Q$ do satisfy the axioms of commutative $\Q$-algebras
(\cf~the proof of~\cref{thm:characterisation:BAC}).
The design of the product rule format
has been guided by the possibility of proving necessity.

%
%


The equational characterisation that we provide is \emph{syntactic},
which should be contrasted with the notion of well-behavedness from~\cite{BorealeGorla:CONCUR:2021},
which is \emph{semantic} 
(and only known to be sufficient).
It immediately follows that being a \BAC~$P$-product is decidable,
which was by no means obvious a priori.
In fact, \cref{thm:characterisation} allows us to infer the following explicit description.
We say that a product rule $P$ is \emph{simple}
if there are rational constants $\alpha, \beta, \gamma \in \Q$ \st
\begin{equation}
    \tag{$\alpha\beta\gamma$}
    \label{eq:simple}
    \left\{\ 
    \begin{aligned}
        P(x, \dot x, y, \dot y)
            &= \alpha \cdot xy + \beta \cdot (x \dot y + \dot x y) + \gamma \cdot \dot x \dot y, \\
            \text{ with } \alpha \cdot \gamma &= \beta \cdot (\beta - 1).
    \end{aligned}
    \right.
\end{equation}
%
%
\begin{corollary}[Classification]
    \label{cor:classification}
    A $P$-product is \BAC\- 
    if, and only if, $P$ is simple.
\end{corollary}
\noindent
For instance, the Leibniz product rule for the shuffle product is simple with $\alpha = 0, \beta = 1, \gamma = 0$.
We find it remarkable that \BAC~$P$-products admit such an elementary description.

\subparagraph*{Second challenge: $P$-automata.}

We resolve the second challenge thanks to the characterisation of \BAC~$P$-products from \cref{thm:characterisation},
which allows us to prove that 
quotienting $P$-automata \wrt~the axioms of commutative $\Q$-algebras respects transitions
(\emph{invariance property}; \cf~\cref{lem:invariance}).
This property is crucial.
Thanks to it, we obtain {\it bona fide} polynomial $P$-automata,
whose states belong to the $\Q$-algebra of commutative polynomials%
\footnote{By $\hompoly \Q X$ we denote polynomials without constant term.
This restriction is necessary since we do not insist that series algebras have a multiplicative identity.}:
\begin{align}
    \label{eq:polynomial algebra}
    \tag{polynomial $\Q$-algebra}
    \tuple{\hompoly \Q X; 0, (c \cdot {\_})_{c \in \Q}, +, \cdot}.
\end{align}
As a consequence, the semantic function of a $P$-automaton
is a homomorphism of commutative $\Q$-algebras~(\cref{lem:polynomial automata:homomorphism}),
achieving full abstraction.
This generalises the result of Boreale and Gorla~\cite[Lemma 3.6 and Theorem 3.7]{BorealeGorla:CONCUR:2021} in three respects:
\begin{inparaenum}[(a)]
    \item we do not require well-behavedness (which is implied by the \BAC~property),
    \item we do not require multiplicative identities, and
    \item we work over arbitrary alphabets.
\end{inparaenum}

While multiplicative identities are tangential in our approach
(we do not need them for the bulk of our results, which are thus stronger),
nonetheless 
we provide a complete characterisation of the \BAC~$P$-products
admitting a multiplicative identity in~\cref{sec:multiplicative identity}.

\subparagraph*{Third challenge: The equivalence problem.}

We observe that ideal compatibility
is satisfied by every simple product rule,
and thus, by \cref{cor:classification}, by 
every \BAC~$P$-product.
This yields decidability of the equivalence problem,
which is our other main contribution.
\begin{restatable}[Decidability]{theorem}{thmDecidability}
    \label{thm:equivalence:decidability}
    The equivalence problem for finite-variable polynomial $P$-automata over finite alphabets is decidable 
    for every \BAC~$P$-product.
\end{restatable}
\noindent
%
This subsumes known decidability for Hadamard~\cite[Theorem 4]{BenediktDuffSharadWorrell:LICS:2017},
shuffle~\cite[Theorem 1]{Clemente:CONCUR:2024},
and infiltration automata~\cite[§5.A]{Clemente:LICS:2025}.
Since Hilbert's theorem relies strongly on the fact that $\hompoly \Q X$
is a finitely generated commutative (non-unital) $\Q$-algebra,
our results indicate a clear limit of Hilbert's method for deciding equivalence of $P$-automata when $P$ is not simple.

\subsection{Additional results}

\subparagraph*{Computational complexity.}
The complexity of the equivalence problem depends on the product rule.
We observe that an Ackermannian upper bound holds in general,
and already for Hadamard automata one cannot do better,
since equivalence is Ackermann-hard~\cite[Theorem~1]{BenediktDuffSharadWorrell:LICS:2017}.
%

\subparagraph*{Commutativity problem.}
Decidability of the equivalence problem has applications to other decision problems.
A series $f \in \series \Q \Sigma$ is \emph{commutative}
if the output $f(w)$ is invariant under permutation of the order of the symbols in the input $w \in \Sigma^*$.
The \emph{commutativity problem} for a (finitely presented) class of series
amounts to decide whether a given series is commutative.
We have studied this problem in~\cite{Clemente:LICS:2025},
where we use it to decide solvability in power series of certain classes of multivariate recursions and differential equations.
The main result in that work is that the commutativity problem is decidable
for Hadamard, shuffle, and infiltration automata~\cite[Theorem 2]{Clemente:LICS:2025}.
We can amply generalise this result in the general framework of $P$-products:
We show in~\cref{app:commutativity} that commutativity of $P$-automata is decidable \emph{for every \BAC~$P$-product}.
This unifies and generalises all results in~\cite{Clemente:LICS:2025}.

\subparagraph*{Mechanical formalisation.}
We have formalised most results of this paper in the
\href{https://github.com/agda/agda}{Agda programming language}~\cite{Norell:AFP:2008}.
This document contains clickable links to the \agda{formalisation}{}.

\subsection{Related work}

The study of products of series via their product rules
is inspired by behavioural differential equations~\cite[Sec.~7]{BasoldHansenPinRutten:MSCS:2017}.
In the special case of a single-letter alphabet,
this program has been advanced by the already mentioned work of Boreale and Gorla,
who studied \emph{$(F, G)$-products}~\cite[Def.~3.1]{BorealeGorla:CONCUR:2021}.
Here, $F$ is a product rule
and $G$ a term modelling the multiplicative identity (i.e., the unit).
The format allowed for $F$ is slightly more expressive than ours,
allowing the authors to also capture the \emph{Cauchy product} over single letter alphabets,
which does not fall in our scope.
We remark that the Cauchy product is \emph{noncommutative} over alphabets of size $\geq 2$,
therefore the fact that it is commutative in the one-letter case can be regarded as an accident.
The added generality necessary to capture this accident, however,
requires the authors to ``assume from the outset'' that the product is~\BAC~\cite[Remark 3.2]{BorealeGorla:CONCUR:2021}.
\Cref{thm:characterisation} shows that, for our more restrictive product rule format,
the \BAC~property can be proved (or disproved) from the product rule itself.

In the special case of a single-letter alphabet,
some classes of $P$-finite series have already been studied:
Hadamard-finite series coincide with \emph{polynomial recursive sequences}\-
\cite{Cadilhac:Mazowiecki:Paperman:Pilipczuk:Senizergues:ToCS:2021},
and shuffle-finite series coincide with
\emph{constructible differentially algebraic power series}~\cite{BergeronReutenauer:EJC:1990};
we are not aware of previous work on infiltration-finite series for a one-letter alphabet.

Over multi-letter alphabets, the main axes (A)--(D) have been developed in~\cite{Clemente:CONCUR:2024}
for the shuffle product,
and in~\cite{Clemente:LICS:2025} for the Hadamard and infiltration products.
Since these products are captured by our characterisation,
all our results on $P$-finite series, $P$-automata, and decidability of equivalence
unify and generalise the corresponding results in these two works.

\subparagraph*{Other algebraic structures.}

The driving criterion for our study has been to obtain broad classes of series (and corresponding automata)
with a decidable zeroness and equality problems.
For this reason, we have chosen to extend the rational series,
for which the zeroness problem is decidable~\cite{Schutzenberger:IC:1961}.
In this respect, it is crucial to work over the field of rational numbers.

One could consider the more general situation where the rational numbers
are replaced by an arbitrary class of semirings
(e.g., commutative semirings, additively idempotent semirings, etc.),
and investigate characterisations of $P$-products endowing the set of series with an analogous semiring structure.
However such an investigation would have to let go of the decidability of the equivalence problem,
since undecidability already holds for weighted automata over the \emph{tropical semiring} $\tuple{\N, \min, +}$~\cite{Krob:ICALP:1992}.

\subparagraph*{Coalgebraic perspective.}
The product rules that we consider fall under the \emph{abstract GSOS} setting of Turi and Plotkin~\cite{TuriPlotkin:LICS:1997}
(\cf~\cite[Sec.~8.2]{HansenKupkeRutten:LMCS:2017} for the case of streams).
This allows one to study $P$-products and $P$-automata in the context of bialgebras~\cite[Sec.~6.3]{Klin:TCS:2011}
and distributive laws of monads over comonads~\cite{BonsangueHansenKurzRot:LMCS:2015}.
The syntax is modelled with the free monad $T^*$ generated by the term signature functor $T$,
and the semantics with the coalgebra of series,
which is the final coalgebra of the weighted automata functor $W := X \mapsto \Q \times X^\Sigma$.
In this view, a $P$-automaton over a set of variables $X$ is a \emph{corecursive equation} $\AA : X \to W T^* X$
(called \emph{$T$-automaton} in \cite[Sec.~3]{Jacobs:AMC:2006}).
Every product rule $P$~determines a \emph{GSOS law} of $T^*$ over $W$,
a distributive law ensuring 
%
a $T^*$-algebra structure on the final coalgebra of series, and
the homomorphism property for the semantics of $P$-automata.

Going further, we can now quotient the free term monad \wrt~the axioms of commutative $\Q$-algebras,
which are the axioms that we want to hold in the series coalgebra.
In this way, we obtain the (nonfree) monad $\mathcal P$ of commutative polynomials over the rationals.
Since the new syntax is not free anymore,
not every product rule $P$ determines a GSOS law of $\mathcal P$ over $W$.
In this context, \cref{thm:characterisation} (together with~\cref{lem:invariance}) provides a necessary and sufficient condition on $P$ for this to happen.
As mentioned in~\cite[Remark 2]{BorealeCollodiGorla:ACMTCL:2024},
sufficiency would follow by showing that the distributive law induced by a \BAC~$P$-product preserves the axioms of commutative $\Q$-algebras,
and then appeal to~\cite[Theorem~4.3]{BonsangueHansenKurzRot:LMCS:2015}.
This is indeed the case, however for clarity we opted for a self-contained presentation.


\subparagraph{\it Organisation.}
The rest of the paper is organised as follows.
In~\cref{sec:preliminaries} we introduce the necessary mathematical preliminaries.
In~\cref{sec:rational series} we use the theory of rational series as an illustrative example of the main axes (A)--(D).
In~\cref{sec:general} we introduce $P$-products and study the theory of $P$-finite series and $P$-automata,
without making any assumption on $P$.
In~\cref{sec:special} we focus on \BAC~$P$-products and we provide the main results of the paper.
Finally, in~\cref{sec:future work} we discuss future work and open problems.
Additional proofs omitted in the main text can be found in~\cref{app:extra}.
A second appendix~\cref{app:commutativity} contains results on the commutativity problem.

\section{Preliminaries}
\label{sec:preliminaries}


Recall that $\Sigma^*$ is the set of finite words over a finite alphabet $\Sigma = \set{a_1, \dots, a_d}$,
which is a monoid under the operation of concatenation
with neutral element the empty word $\e$.
We refer to~\cite{BerstelReutenauer:CUP:2010} for a general introduction to series.
We consider series $\series \Q \Sigma$ over the rationals $\Q$,
however many results in this paper hold for an arbitrary field;
we will highlight which properties of $\Q$ we use when needed.
We use $f, g, h$ to denote series
and write a series $f$ as $\sum_{w \in \Sigma^*} f_w \cdot w$,
where the value of $f$ at $w$ is $f(w) := f_w \in \Q$.
For instance $f = 1 \cdot \e + 1 \cdot a + 1 \cdot aa + \cdots$
is the series over $\Sigma = \set a$ mapping every word to $1$.
The \emph{coefficient extraction} function $\coefficient \_ \_$
maps a word $w$ and a series $f$ to the coefficient $\coefficient w f := f_w$.
%
%
%
We endow series with a dynamic structure:
For every input symbol $a \in \Sigma$, the \emph{left derivative} by $a \in \Sigma$ is the operator
$\deriveleft a : \series \Q \Sigma \to \series \Q \Sigma$ 
mapping a series $f$ to the series $\deriveleft a f$ defined by
\begin{align}
    \label{eq:left derivative}
    \coefficient w {(\deriveleft a f)}
        := \coefficient {a \cdot w} f, \quad \text{for every } w \in \Sigma^*.
\end{align}
Derivatives extend to all finite words homomorphically:
$\deriveleft \e f := f$ and $\deriveleft {a \cdot w} f := \deriveleft w (\deriveleft a f)$
for all $a \in \Sigma$ and $w \in \Sigma^*$. 
This is the series analogue of the left quotient of a language
and provides a convenient way to work with series coinductively.
%
%
For instance, we can use left derivatives to prove that two series are equal:
A \emph{series bisimulation} is a binary relation
$R \subseteq \series \Q \Sigma \times \series \Q \Sigma$
\st~$(f, g) \in R$ implies
\begin{inparaenum}[(B.1)]
    \item\label{bisimulation:1}%
    $f_\e = g_\e$, and
    \item\label{bisimulation:2}%
    for every letter $a \in \Sigma$,
    we have $(\deriveleft a f, \deriveleft a g) \in R$.
\end{inparaenum}
%
%
The following equality principle is proved by induction on the length of words.
\begin{lemma}[\protect{\cite[Theorem 9.1]{Rutten:TCS:2003}}]
    \label{lem:bisimulation}
    If there exists a bisimulation $R$ \st~$(f, g) \in R$, then $f = g$.
\end{lemma}

\begin{table*}
    \begin{center}
        \small
        \makegapedcells
        \begin{tabular}{l|l|l|l}
            operation & notat. & initial value & left derivative by $a$ \\
            \hline\hline
            zero & $\zero$ & $\coefficient \e \zero = 0$ & $\deriveleft a \zero = \zero$ \\
            scalar multiplication by $c \in \Q$
                & $c \cdot f$
                & $\coefficient \e (c \cdot f) = c \cdot \coefficient \e f$
                & $\deriveleft a (c \cdot f) = c \cdot \deriveleft a f$ \\
            addition
                & $f + g$
                & $\coefficient \e (f + g) = \coefficient \e f + \coefficient \e g$
                & $\deriveleft a (f + g) = \deriveleft a f + \deriveleft a g$ \\
            \hline
            Hadamard product (pointwise)
                & $f \hadamard g$
                & $\coefficient \e (f \hadamard g) = \coefficient \e f \cdot \coefficient \e g$
                & $\deriveleft a (f \hadamard g) = \deriveleft a f \hadamard \deriveleft a g$ \\
            shuffle product
                & $f \shuffle g$
                & $\coefficient \e (f \shuffle g) = \coefficient \e f \cdot \coefficient \e g$
                & $\deriveleft a (f \shuffle g) = \deriveleft a f \shuffle g + f \shuffle \deriveleft a g$ \\
            infiltration product
                & $f \infiltration g$
                & $\coefficient \e (f \infiltration g) = \coefficient \e f \cdot \coefficient \e g$
                & $
                    \deriveleft a (f \infiltration g) = \deriveleft a f \infiltration g + f \infiltration \deriveleft a g + \deriveleft a f \infiltration \deriveleft a g
                 $
        \end{tabular}\\[2ex]
    \end{center}
    \caption{Common operations on series, defined coinductively.}
    \label{table:basic operations on series}
\end{table*}

\section{Rational series}
\label{sec:rational series}

In this section, we provide a brief overview of the four axes (A)--(D) in the simple case of rational series (\cref{fig:rational series}).
There will be no new results, however it will be useful for comparison with
the more general setting of $P$-finite series that we will study in~\cref{sec:general}.

\medskip \noindent {\bf (A)}.
The set of series carries a natural $\Q$-vector space structure
given by the point-wise lifting of addition and scalar multiplication over $\Q$~(\cf~\cref{eq:vector space structure}).
To be in line with the rest of the paper,
we present these operations via \emph{behavioural differential equations} (standard in coalgebra~\cite{Rutten:TCS:2003,Rutten:MSCS:2005}),
by specifying their initial value and how their left derivatives behave.
See the upper half of~\cref{table:basic operations on series},
where we define the zero series $\zero$,
scalar multiplication $c \cdot \_$ of a series by a rational $c \in \Q$,
and addition of two series $\_ + \_$.
%

\medskip \noindent {\bf (B)}.
Combining the linear algebraic structure with the dynamic structure given by left derivatives,
we obtain the class of \emph{rational series}.
We begin from an example.
%
\begin{example}[Fibonacci series]
    \label{ex:fibonacci series}
    Consider the univariate series $f, g \in \series \Q a$
    uniquely defined by
    $\coefficient \e f = 0$, 
    $\coefficient \e g = 1$,
    and
    $\deriveleft a f = f + g$,
    $\deriveleft a g = f$.
    Then $\coefficient {a^n} f$ is the $n$-th Fibonacci number, for every $n \in \N$.
    For instance, $\deriveleft {a^3} f = \deriveleft {a^2} (f + g) = \deriveleft a (2 \cdot f + g) = 3 \cdot f + 2 \cdot g$
    and thus $\coefficient {a^3} f = \coefficient \e (\deriveleft {a^3} f) = 3 \cdot 0 + 2 \cdot 1 = 2$.
\end{example}
Generalising the example, we say that a series $f$ is \emph{linearly finite}
if there are finitely many generators $g_1, \dots, g_k$ with $f = g_1$ \st\-
for every $1 \leq i \leq k$ and $a \in \Sigma$,
the left derivative $\deriveleft a g_i$ is a $\Q$-linear combination of the generators $g_1, \dots, g_k$.
This class in fact coincides with the well-known rational series~\cite{BerstelReutenauer:CUP:2010}.
For a single-letter alphabet, they are called \emph{linear recursive sequences}
(\aka~\emph{C-finite sequences}~\cite[Ch.~4]{KauerPaule:Tetrahedron:2011}).
The advantage of the algebraic definition is that it readily provides closure properties with respect to the algebraic and dynamic structure
(\emph{closure lemma}; \cf~\cref{lem:P-finite:closure properties}):
The class of linearly-finite series contains $\zero$
and is effectively closed under scalar multiplication, addition, and left derivatives.

\medskip \noindent {\bf (C)}.
The disadvantage of the algebraic class is that it is not directly amenable to algorithmic manipulation.
For this reason, one introduces an automaton model recognising the same class.
Corresponding to rational series one has weighted automata,
however we present it differently to better align with the rest of the paper.
A \emph{linear term} over a set of variables $X$
is built from constructors corresponding to the vector space operations,
%
\begin{align}
    \label{eq:linear terms}
    \tag{\text{linear terms}}
    u, v ::= x \sep 0 \sep c \cdot u \sep u + v,
    \quad \text{ where } x \in X \text{ and } c \in \Q.
\end{align}
Note that $x + y$, $(x + y) + 0$, and $y + x$ are pairwise different terms.
Since series form a $\Q$-vector space,
we can quotient linear terms by the axioms of $\Q$-vector spaces.
The resulting equivalence classes are homogeneous polynomials over $X$ of degree one, denoted $\Lin X$.
%
%
A \emph{linear automaton} is a tuple $\AA = \tuple{\Sigma, X, F, \Delta}$
where $\Sigma$ is a finite alphabet of \emph{input symbols},
$X$ is a set of \emph{variables},
$F : X \to \Q$ is the \emph{output function},
and $\Delta : \Sigma \to X \to \Lin X$ is the \emph{transition function}.
A linear automaton is \emph{finite-variable} if $X$ is a finite set.
The transition function $\Delta_a$ ($a \in \Sigma$) and the output function $F$
are extended to all terms by linearity. 
Given an initial term $u \in \Lin X$,
the automaton recognises the unique series $\Asem \AA u \in \series \Q \Sigma$
coinductively defined by
\begin{equation*}
        \coefficient \e (\Asem \AA u)
            = F u,
        \quad \text{and} \quad
        \deriveleft a (\Asem \AA u)
            = \Asem \AA {\Delta_a u}, \forall a \in \Sigma.
\end{equation*}
These definitions are best illustrated by an example.
\begin{example}[Fibonacci automaton]
    Consider the linear automaton $\AA = \tuple{\Sigma, X, F, \Delta}$
    over a singleton input alphabet $\Sigma = \set a$ and
    two variables $X = \set{x, y}$ \st\-
    $F x = 0$, $F y = 1$, $\Delta_a x = x + y$, and $\Delta_a y = x$.
    %
    %
    The term $x$ recognises the Fibonacci series from~\cref{ex:fibonacci series}.
    For instance, $\Delta_{a^3} x = \Delta_{a^2} (x + y) = \Delta_a (2 \cdot x + y) = 3 \cdot x + 2 \cdot y$
    and thus $\coefficient {a^3} (\Asem \AA x) = F (3 \cdot x + 2 \cdot y) = 2$.
\end{example}
The semantics of linear automata is a homomorphism from the $\Q$-vector space $\Lin X$
to that of series (\emph{homomorphism lemma}; \cf~\cref{lem:homomorphism}),
thus the linearly-finite series coincide with the series recognised by finite-variable linear automata
(\emph{coincidence lemma}; \cf~\cref{lem:coincidence}).

\medskip \noindent {\bf (D)}.
The advantage of the syntactic model is that it provides an algorithm to decide equality of rational series.
In the case of linear automata, this goes back to the work of Schützenberger~\cite{Schutzenberger:IC:1961}
and relies on linear-algebraic techniques.
We will see how this algorithm generalises to $P$-finite series in~\cref{sec:equivalence algorithm}.

%
%
%

\renewcommand\thesubsection{\thesection(\Alph{subsection})}

\section{$P$-products, $P$-finite series, and $P$-automata}
\label{sec:general}

In this section we introduce a coinductive framework
to define a large family of binary operations on series,
which for simplicity we call \emph{products} (however we make no assumptions about their properties).
Our starting point is that the Hadamard, shuffle, and infiltration products
can be defined coinductively based on a \emph{product rule} $P$ satisfied by their left derivatives;
see the lower half of \cref{table:basic operations on series}.
We remark that products can also be defined by induction on the length of words,
however the inductive definitions quickly become very complicated,
while the coinductive definitions are simple and uniform.
In~\cref{sec:P-products} we capture this situation by the notion of \emph{$P$-product},
observing in~\cref{sec:P-finite series} that it gives rise to a corresponding class of \emph{$P$-finite series}.
In~\cref{sec:P-automata} we define \emph{$P$-automata}
and as the main result of this section
we show in~\cref{lem:coincidence} that a series is $P$-finite iff it is recognised by a finite-variable $P$-automaton.
Such a \emph{coincidence lemma}, which was previously known for the notable products,
in fact holds in the general setting of $P$-products.

\subsection{$P$-products}
\label{sec:P-products}

\subparagraph{Syntax.}
We use \agda{terms}{General/Terms/\#sec:terms} to define product rules satisfied by coinductively defined products.
We extend~\cref{eq:linear terms} over $X$ by adding a new binary constructor for product:
%
%
\begin{align}
    \label{eq:terms}
    \tag{\text{terms}}
    u, v ::= x \sep 0 \sep c \cdot u \sep u + v \sep u * v,
    \quad \text{ where } x \in X \text{ and } c \in \Q.
\end{align}
%
%
Let $\Terms X$ be the set of all terms over $X$.
A term is \emph{linear} if it does not use ``$*$''.
Sometimes we write just $u v$ instead of $u * v$,
$-u$ for $(-1) \cdot u$, and $u - v$ for $u + (-v)$.
%
%
We remark that terms satisfy no nontrivial identities,
for instance $xy$ and $yx$ are different terms.

\subparagraph{Product rules.}
A \emph{product rule} is a term $P$ 
over variables $x, \dot x, y, \dot y$.
%
%
Each product rule $P$ gives rise to a notion of $P$-product of series, as we now explain.

\subparagraph{Semantics of terms and $P$-products.}
\label{sec:semantics of terms}
A \emph{valuation} is a function $\varrho : X \to \series \Q \Sigma$
assigning a series to every variable.
Let $\Val X$ be the set of valuations over $X$.
We simultaneously define a \emph{$P$-product} on series $* : \series \Q \Sigma \to \series \Q \Sigma \to \series \Q \Sigma$
and the \emph{semantics} of terms $\sem {\_}\!{\_} : \Terms X \to \Val X \to \series \Q \Sigma$.
The definition of the product depends on the semantics of terms,
which in turn depends on the product.
For every two series $f, g \in \series \Q \Sigma$,
let $f * g$ be the unique series \st
\begin{align}
    \label{eq:product rule}
    \tag{$*$}
    \coefficient \e (f * g)
        = \coefficient \e f \cdot \coefficient \e g
    \qquad\text{and}\qquad
    \deriveleft a (f * g)
        = P(f, \deriveleft a f, g, \deriveleft a g),
\end{align}
where $P(f, \deriveleft a f, g, \deriveleft a g)$ is a shorthand for $\sem P_\varrho$
for the valuation $\varrho = [x \mapsto f, \dot x \mapsto \deriveleft a f, y \mapsto g, \dot y \mapsto \deriveleft a g]$.
For every $\varrho \in \Val X$, let $\sem {\_}_\varrho : \Terms X \to \series \Q \Sigma$
be defined by
\begin{equation}
    \label{eq:semantics}
    \tag{$\sem \_$}
    \begin{aligned}
        \sem 0_\varrho
            &= \zero,\\ 
        \sem x_\varrho
            &= \varrho x,\\
        \sem {c \cdot u}_\varrho
            &= c \cdot \sem u_\varrho,
    \end{aligned}
    \qquad
    \begin{aligned}
        \sem {u + v}_\varrho
            &= \sem u_\varrho + \sem v_\varrho,\\
        \sem {u * v}_\varrho
            &= \sem u_\varrho * \sem v_\varrho.
    \end{aligned}
\end{equation}
Then \cref{eq:product rule,eq:semantics}
uniquely define both the series product and the semantics of terms;
see~\cite[Sec.~8]{HansenKupkeRutten:LMCS:2017} for a justification when $\card \Sigma = 1$,
and~\cref{rem:product via automata} in general.

\begin{remark}
    One may wonder why we use \emph{terms} $P$ to define $P$-products
    instead of directly using \emph{polynomials}
    (as in done in \cite[Definition 3.1]{BorealeGorla:CONCUR:2021} for example).
    There are two, interrelated reasons.
    The first reason is that we want to be able to define $P$-products which are not necessarily \BAC,
    so that our characterisation~\cref{thm:characterisation:BAC} carries some content.
    The second reason is a consequence of the first one, and it is mathematical:
    Since we do not want to assume a priori any algebraic property of $P$-products in order to be able to define them,
    for the equations~\cref{eq:product rule,eq:semantics} to be well-defined
    we need to consider $P$ from the free algebra of terms.
\end{remark}

By~the base case in \cref{eq:product rule},
constant term extraction is a homomorphism
from the series algebra to the underlying coefficient field $\Q$.
This means that at the constant term the algebra of series
is isomorphic to the algebra of rational numbers.
\begin{agdalemma}[General/Products/\#lem:constant-term-homomorphism-lemma]
    \label{lem:constant term homomorphism}
    For every term $u \in \Terms X$ and valuation $\varrho \in \Val X$,
    we have $\coefficient \e \sem u_\varrho = \sem u_{\coefficient \e \circ \varrho}$.
\end{agdalemma}

\subparagraph{Examples of $P$-products.}

Product rules allow us to define a large family of products of series.
We discuss here some examples.
It is convenient to interpret a product rule $P$ in the context of process algebra,
yielding a recipe prescribing how two processes interact when reading the input letter.
For instance by taking $P = 0$ we get $\deriveleft a (f * g) = 0$,
and thus the two processes vanish upon reading any letter;
consequently, the product is zero on any nonempty input.
As another simple example, for $P = xy$ we have $\deriveleft a (f * g) = f * g$,
which means that the two processes do not read any input at all,
therefore $f * g$ is constantly $f_\e \cdot g_\e$.
If we take $P = \dot x y$, thus $\deriveleft a (f * g) = \deriveleft a f * g$,
then the input is ready only by the left process,
thus $\coefficient w {(f * g)} = f_w \cdot g_\e$ for every $w \in \Sigma^*$.
%

We now discuss more interesting examples.
The rule $P = \dot x \dot y$ yields the Hadamard product,
where both processes read the input in parallel,
and thus $\coefficient w {(f \hadamard g)} = f_w \cdot g_w$ for every $w \in \Sigma^*$.
The shuffle product rule $P = \dot x y + x \dot y$ means that the next input letter
can be read either by the left or by the right process (but not by both);
consequently, $\coefficient w {(f \shuffle g)}$ is obtained by summing up $f_u \cdot g_v$
over all ways of partitioning $w$ into subwords $u$ and $v$~\cite[§6.3]{Lothaire:CUP:1997}.
Infiltration $P = \dot x y + x \dot y + \dot x \dot y$ is similar to shuffle,
however the next input can additionally be read by both processes simultaneously;
thus, the positions of $u, v$ can overlap in $w$~\cite[§6.3]{Lothaire:CUP:1997}.
%
%
%

We can also model less standard products.
Consider two scalars $\beta_1, \beta_2 \in \Q$.
The product rule $P = \beta_1 \cdot \dot x y + \beta_2 \cdot x \dot y$
models a \emph{weighted shuffle product}, whereby the next input letter
is read with weight $\beta_1$ by the first series
and with weight $\beta_2$ by the second one.
For instance,
$\coefficient{ab}{(f*g)} = \beta_1^2 f_{ab} g_\e + \beta_1\beta_2 (f_a g_b + f_b g_a) + \beta_2^2 f_\e g_{ab}$.

\subsection{$P$-finite series}
\label{sec:P-finite series}

Every binary operation on series ``$*$'' (even one not necessarily satisfying a product rule),
naturally induces a class of series.
Fix series $g_1, \dots, g_k \in \series \Q \Sigma$,
which we call \emph{generators}.
Let $A := \algebra \Q {g_1, \dots, g_k}$ be the smallest set of series \st\-
\begin{inparaenum}[(1)]
    \item $\zero \in A$;
    \item $g_1, \dots, g_k \in A$;
    \item for every $f, g \in A$, $c \in \Q$
    we have $c \cdot f, f + g, f * g \in A$.
\end{inparaenum}
A series $f$ is \emph{$*$-finite}
if there are generators $f = g_1, g_2, \dots, g_k$ \st\-
for every $1 \leq i \leq k$ and input symbol $a \in \Sigma$,
the left derivative $\deriveleft a g_i$ is in $\algebra \Q {g_1, \dots, g_k}$.
\begin{wrapfigure}{r}{0.55\textwidth}
    \vspace{-2em} 
    \begin{center}
        \begin{tabular}{l|l}
            $P$-product                      & $P$-finite series \\
            \hline
            Hadamard ``$\hadamard$''         & Hadamard finite~\cite[Sec.~III.C]{Clemente:LICS:2025} \\
            shuffle ``$\shuffle$''           & shuffle finite~\cite[Sec.~2.4]{Clemente:CONCUR:2024} \\
            infiltration ``$\infiltration$'' & infiltration finite~\cite[Sec.~V.A]{Clemente:LICS:2025}
        \end{tabular}
    \end{center}
    \vspace{-2em}
\end{wrapfigure}
If the product ``$*$'' is a $P$-product for some product rule $P$,
then we say that $f$ is \agda{\emph{$P$-finite}}{General/FinitelyGenerated/\#sec:P-finite}.
%
%
%
When we instantiate $P$-finiteness to the notable products,
we obtain classes of series which have been previously considered (table on the right).

\begin{example}[Double exponential]
    \label{ex:double exponential series}
    Consider the Hadamard product.
    Let $f$ be the unique series over $\Sigma = \set a$ \st\-
    $\coefficient \e f = 2$ and $\deriveleft a f = f \hadamard f$.
    %
    %
    It is easy to check that $\coefficient {a^n} f = 2^{2^n}$ for all $n \in \N$.
    Thus, $f$ is Hadamard finite over a single generator. 
\end{example}

\begin{example}[Factorial]
    \label{ex:factorial series}
    Consider the shuffle product.
    Let $f$ be the unique series over $\Sigma = \set a$ \st\-
    $\coefficient \e f = 1$ and $\deriveleft a f = f \shuffle f$.
    %
    %
    Thus $f$ is shuffle finite;
    it can be verified that $\deriveleft {a^n} f = n! \cdot f^{{\shuffle}n}$ for every $n \in \N$,
    and thus $\coefficient {a^n} f = n!$.
\end{example}
%
%
%
%
The notion of $P$-finiteness is designed to ensure basic closure properties,
a testimony of the robustness of this class.%

\begin{agdalemma}[General/FinitelyGenerated/\#sec:P-finite-closure-properties]
    \label{lem:P-finite:closure properties}
    The class of $P$-finite series contains $\zero$
    and is effectively closed under the operations of scalar multiplication,
    addition, product, and left derivatives.
    %
    %
\end{agdalemma}
\noindent
This subsumes analogous results for
Hadamard-finite series and multivariate polynomial recursive sequences~\cite[Lemma~8]{Clemente:LICS:2025},
shuffle-finite series and multivariate constructive differentially algebraic series~\cite[Lemma~10]{Clemente:CONCUR:2024},
and infiltration-finite series~\cite[Sec.~V.A]{Clemente:LICS:2025}.
%
%
\begin{proof}[Proof sketch]
    Let $f, g \in \series \Q \Sigma$ be $P$-finite
    with generators $f_1, \dots, f_k$, resp., $g_1, \dots, g_m$.
    Then, 
    \begin{enumerate}
        \item the series $\zero$ 
        is $P$-finite, with the empty tuple of generators,
        \item $c \cdot f$ is $P$-finite with generators $f_1, \dots, f_k$,
        \item $f + g$ is $P$-finite with generators $f_1, \dots, f_k, g_1, \dots, g_m$,
        \item $f * g$ is $P$-finite with generators $f_1, \dots, f_k, g_1, \dots, g_m$,
        \item $\deriveleft a f$ is $P$-finite with generators $f_1, \dots, f_k$. \qedhere
    \end{enumerate}
    %
\end{proof}

\subsection{$P$-automata}
\label{sec:P-automata}

Every product rule $P$ gives rise to \emph{$P$-automata},
an automaton model recognising series.
The main result of this section is that the class of series recognised by $P$-automata over finitely many variables
coincides with the class of $P$-finite series (\cref{lem:coincidence}).

\subparagraph*{Syntax.}

A \emph{$P$-automaton} is a tuple $\AA = \tuple{\Sigma, X, F, \Delta}$
where $\Sigma$ is a finite alphabet of \emph{input symbols},
$X$ is a set of \emph{variables},
$F : X \to \Q$ is the \emph{output function},
and $\Delta : \Sigma \to X \to \Terms X$ is the \emph{transition function}
mapping every letter $a \in \Sigma$ and variable $x \in X$
to a term $\Delta_a x \in \Terms X$.
A $P$-automaton is \emph{linear} if $\Delta_a x$ is a linear term
for every $a \in \Sigma$ and $x \in X$,
and \emph{finite-variable} if $X$ is a finite set.
The syntax of $P$-automata does not depend on $P$, but its semantics does.


\subparagraph*{Semantics.}

The \emph{$P$-extension} of a function $D : X \to \Terms X$
is the function on all terms $\tilde D : \Terms X \to \Terms X$
defined by the following structural induction:
\begin{equation}
    \label{eq:extension}
    \tag{$P$-ext}
    \begin{aligned}
        \tilde D 0
            &:= 0, \\
        \tilde D x
            &:= D x, \\
        \tilde D (c \cdot \alpha)
            &:= c \cdot \tilde D \alpha,
    \end{aligned}
    \qquad
    \begin{aligned}
        \tilde D (\alpha + \beta)
            &:= \tilde D \alpha + \tilde D \beta,\\
        \tilde D (\alpha * \beta)
            &:= P(\alpha, \tilde D \alpha, \beta, \tilde D \beta).
    \end{aligned}
\end{equation}
In the last equation, $P(\alpha, \tilde D \alpha, \beta, \tilde D \beta)$
is the term obtained from $P(x, \dot x, y, \dot y)$ by the substitution
$[x \mapsto \alpha, \dot x \mapsto \tilde D \alpha, y \mapsto \beta, \dot y \mapsto \tilde D \beta]$.
%
%
The semantics of the automaton is defined by coinduction.
A term $\alpha \in \Terms X$ \emph{recognises} the unique series $\Asem \AA \alpha \in \series \Q \Sigma$ \st 
\begin{equation}
    \label{eq:automata semantics}
    \tag{$\Asem \AA {\_}$}
        \coefficient \e (\Asem \AA \alpha)
            = F \alpha
            \quad\text{and}\quad
        \deriveleft a (\Asem \AA \alpha)
            = \Asem \AA {\tilde \Delta_a \alpha}, \quad \text{for all } a \in \Sigma,
\end{equation}
where $\tilde \Delta_a$ is the $P$-extension of $\Delta_a$,
and $F$ is extended to all terms homomorphically.
By convention, the automaton recognises the series $\Asem \AA {x_1}$,
for a distinguished variable $x_1 \in X$.
\begin{remark}[Inductive definition of the semantics]
    Alternatively, we could have defined the extension of transitions to all finite words homomorphically,
    $\Delta_\e \alpha := \alpha$ and $\Delta_{a \cdot w} \alpha := \Delta_w (\Delta_a \alpha)$,
    and then define the semantics as $\coefficient w (\Asem \AA \alpha) := F(\Delta_w \alpha)$ for all $w \in \Sigma^*$.
    We find the coinductive definition above more elegant in proofs,
    and the inductive one more intuitive in examples. 
\end{remark}

\begin{example}
    \label{ex:term automaton}
    Consider the automaton $\AA = \tuple{\Sigma, X, F, \Delta}$
    over a single-letter alphabet $\Sigma = \set a$
    and a single variable $X = \set x$ defined by
    $F x = 1$ and $\Delta_a x = x * x$.
    For the Hadamard product rule,
    the extension of $\Delta_a$ to all terms is a homomorphism
    and thus $\Delta_{a^n} x$ is the full binary tree of height $n$
    where internal nodes are labelled by ``$*$'' and leaves by $x$.
    For instance, $\Delta_{a^2} x = \Delta_a (x * x) = \Delta_a x * \Delta_a x = (x * x) * (x * x)$.
    In this case, $\Asem \AA x$ is the double exponential series from~\cref{ex:double exponential series}.
    We will present other examples of automata in~\cref{sec:polynomial P-automata}.
\end{example}

In \cref{eq:extension} transitions of $P$-automata are extended to all terms
by mimicking the product rule~\cref{eq:product rule}.
Consequently, the semantics of $P$-automata is a homomorphism.

\begin{agdalemma}[General/Automata/\#homomorphism-lemma][Homomorphism lemma]
    \label{lem:homomorphism}
    The semantics of a $P$-automaton is a homomorphism from terms to series:
    $\Asem \AA \alpha = \sem \alpha_{[x \in X \mapsto \Asem \AA x]}$,
    for every $\alpha \in \Terms X$.
    %
    %
\end{agdalemma}


\begin{figure*}
    \begin{center}
        \begin{tabular}{c|c|c|c}
            \multirow{2}{*}{$P$-product} & \multirow{2}{*}{$P$-automata} & homomorphism lem. & coincidence lem. \\
            & & (\cref{lem:homomorphism}) & (\cref{lem:coincidence}) \\
            \hline
            Hadamard ``$\hadamard$''         & Hadamard automata~\cite[Sec.~III.B]{Clemente:LICS:2025} & \cite[Lemma 5]{Clemente:LICS:2025} & \cite[Lemma 7]{Clemente:LICS:2025} \\
            shuffle ``$\shuffle$''           & shuffle automata~\cite[Sec.~2.2]{Clemente:CONCUR:2024} & \cite[Lemma 8+9]{Clemente:CONCUR:2024} & \cite[Lemma 12]{Clemente:CONCUR:2024} \\
            infiltration ``$\infiltration$'' & infiltration automata~\cite[Sec.~V.A]{Clemente:LICS:2025} & \cite[Sec.~V.A]{Clemente:LICS:2025} & \cite[Sec.~V.A]{Clemente:LICS:2025}
        \end{tabular}
    \end{center}
    \caption{Instances of $P$-automata, the homomorphism lemma, and the coincidence lemma.}
    \label{fig:P-automata}
\end{figure*}

The main result of this section is the coincidence of
$P$-finite series and series recognised by $P$-automata.
It is a direct consequence of the homomorphism lemma~(\cref{lem:homomorphism}).
\begin{agdalemma}[General/Automata/\#sec:coincidence][Coincidence lemma]
    \label{lem:coincidence}
    A series is $P$-finite if, and only if,
    it is recognised by a finite-variable $P$-automaton.
\end{agdalemma}
\noindent
$P$-automata and the corresponding homomorphism and coincidence lemmas
generalise analogous statements from the literature; \cf~\cref{fig:P-automata}.
%
%
\begin{proofsketch}
    For the ``only if'' direction,
    given a $P$-finite series $f \in \series \Q \Sigma$ with generators $f = g_1, \dots, g_k$,
    we construct a $P$-automaton $\AA = \tuple{\Sigma, X, F, \Delta}$.
    Variables are $X := \set{x_1, \dots, x_k}$.
    Let $\varrho : X \to \series \Q \Sigma$ be the valuation
    mapping $x_i$ to $g_i$, for all $i \in \set{1, \dots, k}$.
    %
    %
    By definition, $\delta_a g_i \in \algebra \Q {g_1, \dots, g_k}$,
    and thus there is a term $\alpha_i \in \Terms X$ \st~$\delta_a g_i = \sem {\alpha_i}_\varrho$.
    Define the transition function by $\Delta_a x_i := \alpha_i$.
    Finally, the output mapping is given by the constant terms of the generators $F x_i := \coefficient \e {g_i}$.
    This completes the construction of the $P$-automaton $\AA $.
    Correctness amounts to establish $\Asem \AA {x_1} = \sem {x_1}_\varrho = g_1$.
    More generally, an induction on terms and \cref{lem:homomorphism}
    shows $\Asem \AA \beta = \sem \beta_\varrho$ for every term $\beta \in \Terms X$.

    For the ``if'' direction, we are given a finite-variable $P$-automaton $\AA = \tuple{\Sigma, X, F, \Delta}$
    over variables $X = \set{x_1, \dots, x_k}$.
    Consider the algebra $A := \algebra \Q {g_1, \dots, g_k}$ generated by
    $g_1 := \Asem \AA {x_1}, \dots, g_k := \Asem \AA {x_k}$.
    By the definition of $P$-finiteness, we need to show $\deriveleft a g_i \in A$
    for every $g_i$ and $a \in \Sigma$.
    %
    %
    By the definition of the semantics of automata~\cref{eq:automata semantics},
    we have $\deriveleft a g_i = \deriveleft a (\Asem \AA {x_i}) = \Asem \AA {\Delta_a x_i}$.
    Let $u_i := \Delta_a x_i$.
    By~\cref{lem:homomorphism} we have $\Asem \AA {u_i} = \sem {u_i}_\varrho$
    for the valuation $\varrho := [ x_1 \mapsto g_1, \dots, x_k \mapsto g_k ]$,
    which shows $\deriveleft a g_i \in A$.
\end{proofsketch}

\begin{remark}[The syntactic $P$-automaton]
    \label{rem:product via automata}
    For reasons of presentation,
    we have introduced $P$-products without referring to automata.
    In fact, we can use automata to \emph{define} the product.
    This is done by introducing the \emph{syntactic $P$-automaton} $\SS$,
    which has variables $x_f$ indexed by series $f \in \series \Q \Sigma$,
    output function $F(x_f) := \coefficient \e f$,
    and transitions $\Delta_a x_f := x_{\deriveleft a f}$ ($a \in \Sigma$).
    This is set up so that $\Asem \SS {x_f} = f$. 
    The product of two series $f, g \in \series \Q \Sigma$
    can then be defined as
    $f * g := \Asem \SS {x_f * x_g}$.
    This approach is mathematically more efficient,
    however for reasons of narrative we prefer to present products without referring to automata.
\end{remark}

This concludes the general theory of $P$-products, $P$-finite series, and $P$-automata,
thus developing the axes (A), (B), and (C).
The algorithmic aspect of $P$-automata is much less understood at this level of generality, as we discuss next.

\subsection{The equivalence problem}
\label{sec:equivalence problem}

The \emph{equivalence problem} takes as input two finite-variable $P$-automata $\AA$ and $\BB$
and amounts to decide whether they recognise the same series.
Sometimes equivalence is decidable for trivial reasons.
For instance, for the product rule $P = 0$,
the class of $P$-finite series coincides with the rational series,
for which equivalence is decidable~\cite{Schutzenberger:IC:1961}.
In general, the situation may be more complex.
Since $P$-automata are infinite-state and the term algebra does not have enough structure,
decidability of equivalence is not guaranteed.
In fact, we conjecture that there exists a product rule
\st~equivalence is undecidable.
Nonetheless, for the notable products, equivalence is known to be decidable.
This is based on an algorithmic application of Hilbert's \emph{finite basis theorem}~\cite[Ch.~2, §5, Theorem 4]{CoxLittleOShea:Ideals:2015},
already exploited by Novikov and Yakovenko in the analysis of dynamical systems~\cite{NovikovYakovenko:1999},
where they obtain decidability of equality for univariate ($\Sigma = \set a$) Hadamard and shuffle automata.
This has been the basis of several subsequent decidability results,
such as for multivariate Hadamard automata
(\aka~polynomial automata \cite{BenediktDuffSharadWorrell:LICS:2017}),
multivariate shuffle automata~\cite{Clemente:CONCUR:2024},
and univariate infiltration automata~\cite{BorealeGorla:CONCUR:2021,BorealeCollodiGorla:ACMTCL:2024}.
Remarkably, in all those cases decidability holds for products that are bilinear, associative, and commutative.
We will show in~\cref{sec:equivalence algorithm} that this is not an accident.

\section{Commutative algebras of series}
\label{sec:special}

We restrict our attention to product rules $P$
giving rise to $P$-products which are bilinear, associative, and commutative.
%
%
In~\cref{sec:BAC products} we prove the characterisation (\cref{thm:characterisation})
and classification (\cref{cor:classification}) of such product rules;
%
in~\cref{sec:multiplicative identity} we study multiplicative identities.
The notion of $P$-finite series from~\cref{sec:P-finite series} makes sense for every product rule $P$,
thus there is no corresponding subsection here.
In~\cref{sec:polynomial P-automata} we study \emph{polynomial $P$-automata},
leading in~\cref{sec:equivalence algorithm} to an algorithm for the equivalence problem (\cref{thm:equivalence:decidability}).

\subsection{\BAC~$P$-products}
\label{sec:BAC products}

\subparagraph{Characterisation.}
%
\!\!Series satisfy the \emph{axioms of $\Q$-vector spaces}:
For all $f, g, h \in \series \Q \Sigma, c, d \in \Q$,
\begin{align*}
    \begin{aligned}
        (c \cdot d) \cdot f
            &= c \cdot (d \cdot f), \\
        (c + d) \cdot f
            &= c \cdot f + d \cdot f, \\
        c \cdot (f + g)
            &= c \cdot f + c \cdot g,
    \end{aligned}
    %
    \qquad
    \begin{aligned}
        f + \zero
            &= f, \\
        f + (g + h)
            &= (f + g) + h, \\
        f + g
            &= g + f, \\
    \end{aligned}
    \qquad
    \begin{aligned}
        f - f
            &= \zero.
    \end{aligned}
\end{align*}
The axioms on the left state that scalar multiplication
is compatible with the ring structure of the rationals and addition of series.
Those on the right state that series under addition ``$+$''
form a commutative group with identity $\zero$.
We are interested in products ``$*$''
satisfying the axioms of \emph{commutative $\Q$-algebras},
which are the axioms above together with
%
\begin{align}
    \label{eq:product:additivity}
    \tag{$*$-additivity}
    (f + g) * h
        &= f * h + g * h, \\
    \label{eq:product:homogeneity}
    \tag{$*$-homogeneity}
    (c \cdot f) * g
        &= c \cdot (f * g), \\
    \label{eq:product:associativity}
    \tag{$*$-associativity}
    f * (g * h)
        &= (f * g) * h, \\
    \label{eq:product:commutativity}
    \tag{$*$-commutativity}
    f * g
        &= g * f.
\end{align}
The first two properties state that the product is \emph{left bilinear},
which together with commutativity implies \emph{right bilinearity}.
(Over the field of rationals, additivity implies homogeneity,
however we state these two properties separately for clarity.)
We call these axioms and a product satisfying them \emph{\BAC}.

Since we are interested in commutative $\Q$-algebras of series,
it makes sense to quotient terms accordingly.
To this end, let ``$\approx$'' be the smallest congruence on terms
generated by the \BAC~axioms.
Equivalence classes of terms \wrt~``$\approx$''
are just polynomials (in commuting indeterminates) with no constant term.
We can now define the subclass of product rules we are interested in.
\begin{definition}[Special product rules]
    A product rule $P$ is \emph{special} if it satisfies:
    \begin{align}
        \tag{$P\text{-add}$}
        \label{eq:P-additivity}
        P(x + y, \dot x + \dot y, z, \dot z)
            &\approx P(x, \dot x, z, \dot z) + P(y, \dot y, z, \dot z) \\
        \tag{$P\text{-assoc}$}
        \label{eq:P-associativity}
        P(x, \dot x, yz, P(y, \dot y, z, \dot z))
            &\approx P(xy, P(x, \dot x, y, \dot y), z, \dot z) \\
        \tag{$P\text{-comm}$}
        \label{eq:P-commutativity}
        P(x, \dot x, y, \dot y)
            &\approx P(y, \dot y, x, \dot x).
    \end{align}
\end{definition}
We can now restate the characterisation from~\cref{thm:characterisation}. 
%
\begin{theorem}[Characterisation]
    \label{thm:characterisation:BAC}
    A $P$-product is \BAC~iff the product rule $P$ is special.
    %
\end{theorem}
\noindent
The \agda{``if'' direction}{Special/Products/\#series-algebras} holds over every field,
while the ``only if'' direction requires the field to be infinite. 
This generalises similar observations for specific products
(\eg~univariate Cauchy product~\cite[Proposition 2.5]{BasoldHansenPinRutten:MSCS:2017})
to all products defined by a special product rule.
Since the characterisation is equational and equivalence of polynomials is decidable,
it suggests mechanisation of this kind of arguments.
\begin{proof}[Proof sketch]
    For the ``only if'' direction, assume that ``$*$'' is \BAC\-
    and we show that its product rule $P$ is special.
    Consider~\cref{eq:commutativity} (the other cases are similar).
    Take two arbitrary series $f, g \in \series \Q \Sigma$.
    By commutativity, we have $f * g = g * f$.
    By taking the left derivative of both sides with respect to an arbitrary letter $a \in \Sigma$,
    by the product rule we have
    \begin{align*}
        &\deriveleft a (f * g)
        = P(f, \deriveleft a f, g, \deriveleft a g)
        = P(g, \deriveleft a g, f, \deriveleft a f)
        = \deriveleft a (g * f).
    \end{align*}
    Since constant term extraction is a homomorphism~(\cref{lem:constant term homomorphism}),
    by taking the constant term of both sides we obtain
    %
        $P(\coefficient \e f, \coefficient a f, \coefficient \e g, \coefficient a g) =
        P(\coefficient \e g, \coefficient a g, \coefficient \e f, \coefficient a f)$.
    %
    Since we can choose $\coefficient \e f$, $\coefficient a f$, $\coefficient \e g$, and $\coefficient a g$ arbitrarily,
    the two terms $P(x, \dot x, y, \dot y)$ and $P(y, \dot y, x, \dot x)$ denote the same polynomial function.
    Over an infinite field (such as $\Q$) two polynomial functions are equal
    if, and only if, their corresponding polynomials are equal,
    yielding~\cref{eq:P-commutativity}.
    %
    %

    For the \agda{``if'' direction}{Special/Products/\#series-algebras},
    assume that $P$ is special.
    %
    We prove a more general property.
    \begin{claim*}
        For every term $u, v \in \Terms X$ and valuation $\varrho$,
        if $u \approx v$ then $\sem u_\varrho = \sem v_\varrho$.
        %
    \end{claim*}
    \noindent
    The claim implies that ``$*$'' is \BAC\-
    by considering the valuation $\varrho := [x \mapsto f, y \mapsto g, z \mapsto h]$
    and the equivalences $(x + y) * z \approx x * z + y * z$,
    $x * (y * z) \approx (x * y) * z$, resp., $x * y \approx y * x$.
    
    We now prove the claim.
    Consider the relation ``$\sim$'' on series defined as follows:
    \begin{align}
        \text{For every $\varrho : X \to \series \Q \Sigma$:}
        \qquad\qquad
        \sem u_\varrho \sim \sem v_\varrho
            \quad\text{iff}\quad
                u \approx v.
    \end{align}
    We show that this is a bisimulation, which suffices to prove the claim by~\cref{lem:bisimulation}.
    Assume $\sem u_\varrho \sim \sem v_\varrho$.
    Since $u \approx v$ and term equivalence respects the underlying polynomial function,
    by~\cref{lem:constant term homomorphism} the constant terms agree,
    %
        $\coefficient \e \sem u_\varrho
        = u(\coefficient \e \circ \varrho)
        = v(\coefficient \e \circ \varrho)
        = \coefficient \e \sem v_\varrho$.
    %
    We now show that ``$\sim$'' is preserved by left derivatives.
    %
    Since ``$\approx$'' is the smallest congruence generated by the \BAC~axioms 
    and the semantics is a homomorphism by definition~\cref{eq:semantics},
    it suffices to show preservation for each axiom.
    %
    %
    For instance, we consider~\cref{eq:product:commutativity}.
    For two terms $u, v$ we need to show
    $\deriveleft a \sem {u * v}_\varrho \sim \deriveleft a \sem {v * u}_\varrho$.
    We have
    \begin{align}
        \nonumber
        \deriveleft a \sem {u * v}_\varrho
        &= \deriveleft a (\sem u_\varrho * \sem v_\varrho) =
        && \text{(by~\cref{eq:semantics})} \\
        \label{eq:coinductive assumption}
        \tag{!}
        &= \sem {P(x, \dot x, y, \dot y)}_\eta \sim 
        && \text{(by~\cref{eq:commutativity})}  \\
        \nonumber
        &\sim \sem{P(y, \dot y, x, \dot x)}_\eta =
        && \text{(by~\cref{eq:semantics})} \\
        \nonumber
        &= \deriveleft a (\sem v_\varrho * \sem u_\varrho)
        = \deriveleft a \sem {v * u}_\varrho,
    \end{align}
    where $\eta$ is the valuation
    $\eta := [x \mapsto \sem u_\varrho, \dot x \mapsto \deriveleft a \sem u_\varrho,
    y \mapsto \sem v_\varrho, \dot y \mapsto \deriveleft a \sem v_\varrho]$.
    We have invoked the coinductive assumption in~\cref{eq:coinductive assumption},
    where the use of~\cref{eq:commutativity} coinductively relies on all the axioms of commutative $\Q$-algebras.
    The other cases are similar.
\end{proof}

\subparagraph{Classification.}

The characterisation from~\cref{thm:characterisation:BAC} puts strong constraints on special product rules,
allowing us to obtain an explicit description for them (\cref{cor:classification}).
Consider an arbitrary product rule $P$ of degree $\leq d$,
\begin{align*}
    P =
        \sum_{i + j + k + \ell \leq d}
            \alpha_{i, j, k, \ell} \cdot
            x^i \dot x^j y^k \dot y^\ell,
            \qquad \alpha_{i, j, k, \ell} \in \Q.
\end{align*}
We show that if $P$ is special, then it is simple~\cref{eq:simple}.
We proceed by progressively imposing bilinearity, associativity, and commutativity, in this order.

\subparagraph*{\it Bilinearity.}
By left additivity~\cref{eq:left additivity}
we have $\alpha_{i, j, k, \ell} = 0$ unless $i + j = 1$.
%
%
Similarly, by right additivity $k + \ell = 1$.
By renaming the coefficients,
we can write $P$ in the \emph{bilinear form}
\begin{align}
    \label{eq:bilinear form}
    \tag{bilinear~form}
    P =
        \alpha \cdot xy +
        \beta_1 \cdot x\dot y +
        \beta_2 \cdot \dot xy + 
        \gamma \cdot \dot x\dot y.
\end{align}

\begin{example}[Lamperti product]
    If in~\cref{eq:bilinear form} we put $\alpha = 0$,
    %
    %
    then we get the multivariate \emph{Lamperti product}~\cite{Lamperti:AMM:1958} (\cf~also~\cite[footnote 2]{Fliess:1974}).
    %
    This product is bilinear, however in general it is neither associative nor commutative.
    The Hadamard,
    shuffle,
    and infiltration products
    are special cases.
\end{example}

\subparagraph*{\it Associativity.}
By imposing~\cref{eq:associativity} to a product rule in~\cref{eq:bilinear form},
we obtain the following four identities on the parameters $\alpha, \beta_1, \beta_2, \gamma \in \Q$
(details in~\cref{app:associativity}):
\begin{equation}
    \label{eq:associativity identities}
        \alpha \cdot (\beta_1 - \beta_2) = 0, \quad
        \gamma \cdot (\beta_1 - \beta_2) = 0, \quad
        \alpha \cdot \gamma = \beta_1 \cdot (\beta_1 - 1), \quad
        \alpha \cdot \gamma = \beta_2 \cdot (\beta_2 - 1).
\end{equation}

\subparagraph*{\it Commutativity.}
Finally, by imposing~\cref{eq:commutativity} in~\cref{eq:associativity identities},
we have $\beta := \beta_1 = \beta_2$~and we obtain the simple form announced in the introduction~\cref{eq:simple}.
This concludes the proof of~\cref{cor:classification}.
(Note that~\cref{eq:associativity identities} does not force $\beta_1 = \beta_2$;
if $\beta_1 \neq \beta_2$, then $\alpha = \gamma = 0$, $\set {\beta_1, \beta_2} = \set {0, 1}$,
and thus the only bilinear-associative-\emph{non}commutative $P$-products
are $P = x \dot y$ and $P = \dot x y$.)
%
%

\subsubsection{Multiplicative identity}
\label{sec:multiplicative identity}
We investigate the existence of multiplicative identities for simple products.
We say that a simple product rule is \emph{degenerate}
if $\beta = \gamma = 0$, and \emph{nondegenerate} otherwise.
\begin{restatable}[Multiplicative identity]{theorem}{multIden}
    \label{thm:multiplicative identity}
    A simple product has a multiplicative identity $\one$ if, and only if, it is nondegenerate.
    When this is the case, there exists a constant $\eta \in \Q$
    \st~$\one$ can be defined coinductively by
    \begin{align}
        \tag{multiplicative identity}
        \label{eq:multiplicative identity}
        \coefficient \e \one = 1
        \quad\text{and}\quad
        \derive a \one = \eta \cdot \one.
    \end{align}
\end{restatable}

\begin{wrapfigure}{r}{0.37\textwidth}
    \vspace{-1.5em} 
    \begin{tabular}{l|c|c|c||c}
        $P$-product                      & $\alpha$ & $\beta$ & $\gamma$ & $\eta$ \\
        \hline
        Hadamard ``$\hadamard$''         & $0$      & $0$     & $1$ & 1 \\
        shuffle ``$\shuffle$''           & $0$      & $1$     & $0$ & 0 \\
        infiltration ``$\infiltration$'' & $0$      & $1$     & $1$ & 0
    \end{tabular}
    \vspace{-3em}
\end{wrapfigure}
\noindent
The notable products are simple with multiplicative identities, as shown on the right.

\setcounter{subsection}{2}
\subsection{Polynomial $P$-automata}
\label{sec:polynomial P-automata}

%
%
In this section we show that the state space of $P$-automata for a special product rule $P$
can be quotiented \wrt~the \BAC~axioms,
obtaining \emph{polynomial $P$-automata}.
This will be exploited in~\cref{sec:equivalence algorithm} to show that 
equivalence of finite-variable $P$-automata is decidable.

\subparagraph{Invariance property.}
We show that $P$-extensions preserve term equivalence.
This is a key technical result.
\begin{agdalemma}[Special/Automata/\#invariance][Invariance]
    \label{lem:invariance}
    Let $P$ be a special product rule
    and consider a function $D : X \to \Terms X$.
    Recall that $\tilde D : \Terms X \to \Terms X$ is the $P$-extension of $D$ to all terms (\cf~\cref{eq:extension}).
    For terms $u, v \in \Terms X$,
    if $u \approx v$ then $\tilde D u \approx \tilde D v$.
\end{agdalemma}
\begin{proof}[Proof sketch]
    We show invariance for each~\BAC~axiom.
    As an example, we discuss commutativity of multiplication.
    For two terms $u, v \in \Terms X$
    we need to show $\tilde D (u * v) \approx \tilde D (v * u)$:
    \begin{align*}
        \tilde D (u * v)
        &= P(u, \tilde D u, v, \tilde D v)
            && \text{(by def.~of $P$-extension)} \\
        &\approx P(v, \tilde D v, u, \tilde D u)
            && \text{(by~\cref{eq:P-commutativity})} \\
        &= \tilde D (v * u).
            && \text{(by def.~of $P$-extension)}
        \qedhere
    \end{align*}
\end{proof}
By~\cref{lem:invariance}, quotienting the state space of a $P$-automaton by term equivalence respects transitions,
which immediately implies invariance of the semantics.
\begin{agdacorollary}[Special/Automata/\#semantic-invariance][Semantic invariance]
    \label{cor:semantic invariance}
    Let $P$ be a special product rule.
    If $u \approx v$ then, for every $P$-automaton $\AA$,
    we have $\Asem \AA u = \Asem \AA v$.
\end{agdacorollary}

\begin{remark}
    We can use~\cref{cor:semantic invariance} applied to the syntactic $P$-automaton from \cref{rem:product via automata}
    to give an alternative proof of the ``if'' direction of~\cref{thm:characterisation:BAC}.
    In fact, the two proofs are closely related.
    For instance, regarding commutativity,
    by~\cref{cor:semantic invariance} applied to $x_f * x_g \approx x_g * x_f$ we have
    $f * g = \Asem \SS {x_f * x_g} = \Asem \SS {x_g * x_f} = g * f$.
    More generally, in a mathematically minimal presentation we could have presented automata first,
    however we have opted for a direct proof in~\cref{thm:characterisation:BAC} for reasons of narrative.
\end{remark}

By~\cref{lem:invariance,cor:semantic invariance},
we can quotient the state space of a $P$-automaton
and obtain a well-defined \emph{polynomial} $P$-automaton
recognising the same series.
%
%
We find it instructive to define the latter model directly.
In order to do so, we rely on a nontrivial \emph{extension property},
which is interesting to study in its own right.

\subparagraph{Extension property.}

In this section, we use invariance to derive the following nontrivial extension property.
We will use it in~\cref{sec:def polynomial P-automata} to provide a well-defined semantics for polynomial $P$-automata.
By $\hompoly \Q X$ we denote the set of polynomials $p$
with no constant term, i.e., $p(0, \dots, 0) = 0$.
\begin{corollary}[Extension property]
    \label{cor:extension property}
    Let $P$ be a special product rule.
    Every function $D : X \to \hompoly \Q X$
    extends uniquely to a linear function $\tilde D : \hompoly \Q X \to \hompoly \Q X$ \st
    \begin{equation}
            \tilde D (\alpha \beta)
                = P(\alpha, \tilde D \alpha, \beta, \tilde D \beta),
                \quad \forall \alpha, \beta \in \hompoly \Q X.
    \end{equation}
\end{corollary}
\begin{proof}
    The proof is very simple.
    Let $D : X \to \hompoly \Q X$ be a function.
    By forgetting the algebraic structure of polynomials,
    we consider it as a function $D : X \to \Terms X$,
    where $D x$ is now \emph{any} term representing the original polynomial.
    Let $\tilde D : \Terms X \to \Terms X$ be its $P$-extension \emph{on terms}.
    By invariance (\cref{lem:invariance}),
    $\tilde D$ preserves term equivalence,
    and thus it induces a well-defined function on polynomials
    $\tilde D : \hompoly \Q X \to \hompoly \Q X$.
\end{proof}
The extension property subsumes and explains known observations from the literature.
%
\begin{example}[Homomorphic extension]
    Consider the Hadamard product rule $P = \dot x \dot y$.
    The $P$-extension of $D : X \to \hompoly \Q X$
    is the unique homomorphic extension to all homogeneous polynomials.
    In other words, $\tilde D$ acts by simultaneous substitution of $x \in X$ by $D x$.
    For instance, if $D x = x^2$,
    then $\tilde D (x^4) = (D x)^4 = x^8$.
\end{example}

\begin{example}[Derivation extension]
    Consider the shuffle product rule $P = x \dot y + \dot x y$.
    A \emph{derivation} (of the polynomial ring) is a linear endomorphism $D$ on $\hompoly \Q X$
    satisfying the Leibniz rule $D(\alpha \beta) = \alpha \cdot D\beta + D\alpha \cdot \beta$.
    By~\cref{cor:extension property}, every $D : X \to \hompoly \Q X$ extends to a unique derivation.
    We thus recover a basic result in differential algebra~\cite[Ch. 1, Sec. 2, Ex. 4]{Kaplansky:DA:1957}.
    For example, if $D x = x^2$ then $\tilde D (x^4) = 4 \cdot x^3 \cdot D x = 4 \cdot x^5$.
\end{example}

\begin{example}[Infiltration extension]
    Consider the infiltration product rule $P = xy + x \dot y + \dot x y$.
    An \emph{infiltration} is a linear endomorphism $D$ on $\hompoly \Q X$ satisfying the infiltration rule
    $D(\alpha \beta) = \alpha \cdot D\beta + D\alpha \cdot \beta + D\alpha \cdot D\beta$.
    By \cref{cor:extension property}, every $D : X \to \hompoly \Q X$
    extends to a unique infiltration.
    This has been employed in~\cite[Sec.~E.B]{Clemente:arXiv:LICS:2025} (without proof).
    Thanks to the extension property we can avoid a quite tedious direct proof.
\end{example}

It is remarkable that we can obtain the examples above with a single, simple proof,
a testimony of the usefulness invariance (\cref{lem:invariance}, used to prove~\cref{cor:extension property}).

\subsubsection{Syntax and semantics of polynomial $P$-automata.}
\label{sec:def polynomial P-automata}
Let $P$ be a special product rule.
A \emph{polynomial $P$-automaton} is a tuple $\AA = \tuple{\Sigma, X, F, \Delta}$
where $\Sigma$ is a finite alphabet of input symbols,
$X$ is a set of variables,
$F : X \to \Q$ is the output function,
and $\Delta : \Sigma \to X \to \hompoly \Q X$ is the transition function.
By \cref{cor:extension property}, $\Delta_a$ ($a \in \Sigma)$ extends uniquely
to all polynomials $\tilde \Delta_a : \hompoly \Q X \to \hompoly \Q X$.
The semantic function $\Asem \AA {\_} : \hompoly \Q X \to \series \Q \Sigma$
is then defined coinductively as in~\cref{eq:automata semantics}.
By~\cref{lem:homomorphism}, it is a homomorphism of commutative $\Q$-algebras.
\begin{lemma}
    \label{lem:polynomial automata:homomorphism}
    Let $P$ be a special product rule.
    The semantics of polynomial $P$-automata is a homomorphism
    from the~\cref{eq:polynomial algebra} to the~\cref{eq:series algebra}.
\end{lemma}

\begin{wrapfigure}{r}{0.7\textwidth}
    \vspace{-2.5em} 
    \begin{center}
        \begin{tabular}{l|c}
            $P$-product                      & (finite-variable) polynomial $P$-automata \\
            \hline
            Hadamard ``$\hadamard$''         & Hadamard automata~\cite{BenediktDuffSharadWorrell:LICS:2017} \\
            shuffle ``$\shuffle$''           & shuffle automata~\cite{Clemente:CONCUR:2024} \\
            infiltration ``$\infiltration$'' & infiltration automata~\cite{Clemente:LICS:2025}
        \end{tabular}
    \end{center}
    \vspace{-4em}
\end{wrapfigure}
In this way, we recover previous notions of automata,
shown on the right.

\vspace{1.5em}
\begin{example}
    \label{ex:polynomial automata}
    Consider the automaton $\AA = \tuple{\Sigma, X, F, \Delta}$
    over a single-letter alphabet $\Sigma = \set a$
    and a single variable $X = \set x$ defined by
    $F x = 1$ and $\Delta_a x = x^2$.
    For the Hadamard product rule,
    the extension of $\Delta_a$ to all terms is a homomorphism
    and we have $\Delta_{a^n} x = x^{2^n}$.
    In this case, $\Asem \AA x$ is the double exponential series from~\cref{ex:double exponential series}.

    For the shuffle product rule,
    the extension of $\Delta_a$ to all terms is a derivation
    and thus $\Delta_{a^n} x = n! \cdot x^{n+1}$.
    Indeed, this holds for $n = 0$ and inductively we have
    $\Delta_{a^{n+1}} x
        = \Delta_a (\Delta_{a^n} x)
        = \Delta_a (n! \cdot x^{n+1})
        = n! \cdot (n+1) \cdot x^n \cdot \Delta_a x
        = (n + 1)! \cdot x^{n + 2}$.
    Thus, $\Asem \AA x$ is the factorial series from~\cref{ex:factorial series}.
\end{example}


by introducing a fresh variable.




\subsection{The equivalence algorithm}
\label{sec:equivalence algorithm}

In this section, fix a finite alphabet $\Sigma$.
We show that equivalence of finite-variable $P$-automata (and thus $P$-finite series)
is decidable for \emph{every special product rule $P$}.
This is remarkable given the diversity of products covered by our characterisation.
%
%
We could replace the rationals by an arbitrary computable field,
however we stick to $\Q$ for simplicity.
Since equality $f = g$ effectively reduces to \emph{zeroness} $f - g = \zero$,
it suffices to show that the zeroness problem is decidable for $P$-finite series.
Thus fix a finite-variable $P$-automaton $\AA = \tuple{\Sigma, X, F, \Delta}$
and we show how to decide $\Asem \AA {x_1} = \zero$.

For a set of polynomials $Q \subseteq \poly \Q X$,
let the \emph{polynomial ideal generated by $Q$}, written $\ideal Q$,
be the set of all polynomials of the form
\begin{align*}
    \ideal Q :=
    \beta_1 \cdot \alpha_1 + \cdots + \beta_n \cdot \alpha_n,
    \quad\text{with } \alpha_i \in Q, \beta_j \in \poly \Q X.
\end{align*}
%
%
(See~\cite{CoxLittleOShea:Ideals:2015} for more details on algebraic geometry.)
This is a generous over-approximation of $Q$ which is sound for the zeroness problem,
in the sense that if all generators in $Q$ recognise the zero series,
then so do all polynomials in the ideal $\ideal Q$ they generate;
it is a generalisation of the notion of linear span for vector spaces.
%
%
%
For every $n \in \N$,
consider the polynomial ideal $$I_n := \idealof {\Delta_w x_1} {w \in \Sigma^{\leq n}}$$
generated by all polynomials reachable from $x_1$
by reading words of length $\leq n$.
%
%
The following technical lemma describes the relevant properties of the ideals $I_n$.
\begin{restatable}{lemma}{lemIdeals}
    \label{lem:ideals}
    \begin{inparaenum}
        \item $I_n \subseteq I_{n+1}$.
        \item $\Delta_a I_n \subseteq I_{n+1}$.
        \item $I_{n+1} = I_n + \idealof {\Delta_a I_n} {a \in \Sigma}$.
        \item $I_n = I_{n+1}$ implies $I_n = I_{n+1} = I_{n+2} = \cdots$.
    \end{inparaenum}
\end{restatable}
\noindent
The second point in the lemma relies on the fact that for every special product rule $P$,
the polynomial $P$ is in the ideal generated by $y, \dot y$:
\begin{align}
    \label{eq:ideal compatibility}
    \tag{ideal compatibility}
    P \in \ideal {y, \dot y}.
\end{align}
This is a consequence of~\cref{eq:bilinear form}.
This condition was already employed
in the case of a single-letter alphabet $\Sigma = \set a$~\cite[Theorem~4.2]{BorealeCollodiGorla:ACMTCL:2024},
however we \emph{do not assume}~\cref{eq:ideal compatibility},
as done in~\cite{BorealeCollodiGorla:ACMTCL:2024},
but rather we observe that it holds for every special product.
The rest of the argument is standard.
Consider the chain of ideals
\begin{align}
    \label{eq:ideal chain}
    I_0 \subseteq I_1 \subseteq \cdots \subseteq \poly \Q X.
\end{align}
By \emph{Hilbert's finite basis theorem}~\cite[Theorem 4, §5, Ch. 2]{CoxLittleOShea:Ideals:2015},
there is a (minimal) \emph{stabilisation index} $N \in \N$ \st~$I_N = I_{N+1} = \cdots$.
Thanks to the last point of~\cref{lem:ideals}, the stabilisation index is computable:
Indeed, $N$ can be taken to be the smallest $n \in \N$ \st~$I_n = I_{n+1}$,
and ideal equality can be decided by checking that the generators of $I_{n+1}$ are all in $I_n$.
The latter test is an instance of the \emph{ideal membership problem},
which is decidable (in exponential space)~\cite{Mayr:STACS:1989}.
It is crucial that such an $N$ not only exists
(as follows already from Hilbert's theorem, whose proof is notoriously non-constructive),
but that it can be computed
(a consequence of the way the ideals $I_n$'s are constructed).
Once we have computed $N$, we can apply the following decidable criterion for zeroness.
\begin{restatable}{lemma}{lemZeronessCharacterisation}
    \label{lem:zeroness characterisation}
    Let $N \in \N$ be the stabilisation index
    of the ideal chain~\cref{eq:ideal chain}.
    We then have
    \begin{align*}
        \Asem \AA {x_1} = \zero
        \quad \text{iff} \quad
        \coefficient w (\Asem \AA {x_1}) = 0
        \text{ for all } w \in \Sigma^{\leq N}.
    \end{align*}
\end{restatable}
\noindent
Thanks to~\cref{lem:zeroness characterisation},
to decide zeroness we check that the automaton outputs zero for all words shorter than the stabilisation index.
This concludes the proof of decidability of the zeroness and equivalence problems (\cref{thm:equivalence:decidability}).
We stress that this subsumes known decidability results for Hadamard automata~\cite[Theorem 4]{BenediktDuffSharadWorrell:LICS:2017},
shuffle automata~\cite[Theorem 1]{Clemente:CONCUR:2024},
and infiltration automata~\cite[§5.A]{Clemente:LICS:2025}.
Thanks to our characterisation, $P$-automata for a special product $P$
are a natural maximal class of automata
for which zeroness is decidable via Hilbert's method.

\begin{remark}[Complexity]
    \label{rem:complexity of equivalence}
    Zeroness is Ackermann-complete for Hadamard automata
    over an alphabet of size $\card \Sigma \geq 2$~\cite[Theorem~1]{BenediktDuffSharadWorrell:LICS:2017}.
    The Ackermannian upper bound in fact holds for \emph{every} special product.
    This can be obtained by noticing that the upper-bound arguments on the stabilisation index
    in~\cite[Sec.~3]{NovikovYakovenko:1999}
    or~\cite[Sec.~V]{BenediktDuffSharadWorrell:LICS:2017}
    depend only on the fact that the degrees of the generators of $I_n$ are exponential in $n$.
    The latter is true for every $P$-automaton.
    %
\end{remark}

\renewcommand\thesubsection{\thesection.\arabic{subsection}}

\section{Future work}
\label{sec:future work}

\newcommand{\shrink}{\vspace{-2ex}}

In this paper we have studied coinductive products of series satisfying a product rule.
We have provided a complete and decidable description of all product rules $P$ yielding commutative $\Q$-algebras of series,
and when this is the case we have shown that equivalence of $P$-automata is decidable.
This unifies and generalises known results from the literature on the Hadamard, shuffle, and infiltration products.
We outline several directions to extend this work.

\shrink
\subparagraph*{Novel applications of special products.}
While our characterisation applies to infinitely many product rules
(in fact, the set of tuples of parameters $\tuple{\alpha, \beta, \gamma}$
satisfying~\cref{eq:simple} is an \emph{algebraic variety}),
it remains to see whether such rules have convincing applications,
for example in process algebra or in combinatorics.
This level of generality has not (yet) lead to the discovery
of a novel and useful product of series beyond the known Hadamard, shuffle, and infiltration products.

\shrink
\subparagraph{Input-dependent product rules.}
We have considered product rules that are independent of the input symbol. 
From a modelling perspective, it may be desirable to consider \emph{input-dependent product rules} $P_a$ (one for each $a \in \Sigma$).
For instance, in the context of process algebras,
one could use the Hadamard rule $P_a = \dot x \dot y$ for symbols $a$ where we require both processes to synchronise,
the shuffle rule $P_b = \dot x y + x \dot y$ for symbols $b$ where exactly one process makes a step,
and the infiltration rule $P_c = \dot x y + x \dot y + \dot x \dot y$ for symbols $c$
where both processes may make a step independently or synchronously.
Such an extension does not pose significant technical challenges,
and we expect all results to carry over.

\shrink
\subparagraph{Equivalence problem for non-special products.}
Regarding decidability of equivalence,
it would be interesting to find a non-special product rule $P$ \st\-
%
$P$-finite series are strictly more expressive than the rational series and
equivalence of $P$-automata is decidable.
%
Conversely, it would be interesting to find a (necessarily non-special) product rule $P$
\st~the equivalence problem for $P$-automata is undecidable.

\shrink
\subparagraph{Concatenation product and weighted context-free grammars.}
We would like to model the concatenation (Cauchy) product
via more general product rule formats.
This product is bilinear and associative, but not commutative
over alphabets of size $\card \Sigma \geq 2$.
It is well-known that it can be uniquely defined by
$\coefficient \e {(f * g)} = f_\e \cdot g_\e$
and the \emph{Brzozowski product rule}
\begin{align}
    \label{eq:brzozowski product rule}
    \derive a (f * g)
        = (\derive a f) * g + (\coefficient \e f) \cdot \derive a g.
\end{align}
This suggests considering an extended term language
allowing for an operation of extracting the constant term of a series,
and a scalar multiplication.
While~\cref{eq:P-additivity,eq:P-associativity} are still sufficient
to ensure that the product is bilinear and associative,
it is not clear anymore whether they are also necessary.

Another challenge in this direction would be to show decidability of equivalence
for the corresponding \emph{Cauchy-finite series} and \emph{Cauchy automata},
since this is nothing else than equivalence of weighted context-free grammars.
The latter, originating in Knuth's multilanguage semantics for context-free grammars~(\cf~\cite{Knuth:1991}),
is a long-standing open problem\- 
(\cf~\cite{Kuich:multiplicity:1994,ForejtJancarKieferWorrell:IC:2014}).
Note that, in the case of weighted grammars over \emph{commuting} input symbols,
equivalence is decidable with polynomial space complexity,
and even in the counting hierarchy~\cite{BalajiClementeNosanShirmohammadiWorrell:LICS:2023}.

\shrink
\subparagraph{Semirings.}
One could consider a more general setting
where the coefficients come from a class of semirings $\mathcal S$,
rather than from the field of rational numbers $\Q$.
A natural question is to understand for which semiring classes $\mathcal S$
one can characterise $P$-products endowing the set of series with \emph{the same} semiring structure $\mathcal S$.
This amounts to understanding how algebraic properties transfer
from the coefficient semiring to the product of series.

\bibliography{literature}


\appendix
\newpage

\section{Additional proofs}
\label{app:extra}

In this section we present some additional proofs which were omitted from the main text.

\subsection{Associativity}
\label{app:associativity}

We show that a bilinear product rule $P$~\cref{eq:bilinear form}
is associative if, and only if, the parameters $\alpha, \beta_1, \beta_2, \gamma \in \Q$
satisfy the identities from~\cref{eq:associativity identities}.
To this end, consider $P$ in~\cref{eq:bilinear form}.
We now compute both sides of the associativity condition~\cref{eq:P-associativity}:
%
\begin{align*}
    P(x, \dot x, y z, P(y, \dot y, z, \dot z)) =
    \ & \alpha \cdot x y z +
    \beta_1 \cdot x \left(
        \alpha \cdot yz +
        \beta_1 \cdot y \dot z +
        \beta_2 \cdot \dot y z +
        \gamma \cdot \dot y \dot z
    \right) +
    \beta_2 \cdot \dot x y z + {} \\
    & {} + \gamma \cdot \dot x \left(
        \alpha \cdot yz +
        \beta_1 \cdot y \dot z +
        \beta_2 \cdot \dot y z +
        \gamma \cdot \dot y \dot z
    \right),
    \text{ and } \\
    P(xy, P(x, \dot x, y, \dot y), z, \dot z) = 
    \ & \alpha \cdot x y z +
    \beta_1 \cdot x y \dot z +
    \beta_2 \cdot \left(
        \alpha \cdot xy +
        \beta_1 \cdot x \dot y +
        \beta_2 \cdot \dot x y +
        \gamma \cdot \dot x \dot y
    \right) z + {} \\
    & {} + \gamma \cdot \left(
        \alpha \cdot xy +
        \beta_1 \cdot x \dot y +
        \beta_2 \cdot \dot x y +
        \gamma \cdot \dot x \dot y
    \right) \dot z.
\end{align*}
By taking the difference of these two polynomials and gathering like terms, we have
\begin{align*}
    &(\alpha + \beta_1 \alpha - \alpha - \beta_2 \alpha) \cdot xyz +
    (\beta_2 + \gamma\alpha - \beta_2^2)  \cdot \dot x y z +
    (\beta_1 \beta_2 - \beta_2 \beta_1) \cdot x \dot y z +
    (\beta_1^2 - \beta_1 - \gamma \alpha) \cdot x y \dot z + {} \\
    &(\gamma \beta_2 - \beta_2 \gamma) \cdot \dot x \dot y z +
    (\gamma \beta_1 - \gamma \beta_2) \cdot \dot x y \dot z +
    (\beta_1 \gamma - \gamma \beta_1) \cdot x \dot y \dot z +
    (\gamma^2 - \gamma^2) \cdot \dot x \dot y \dot z = 0.
\end{align*}
By equating each coefficient to zero and simplifying,
we obtain~\cref{eq:associativity identities}.

\subsection{Multiplicative identities}
\label{app:multiplicative identities}

We recall and prove the characterisation of multiplicative identities for simple products from~\cref{sec:multiplicative identity}.

\multIden*

\begin{proof}
    Assume that the simple product is nondegenerate.
    We claim that there exists a constant $\eta \in \Q$
    \st~the unique series $\one$ satisfying~\cref{eq:multiplicative identity}
    is a multiplicative identity.
    In other words, we need to find $\eta$ \st
    \begin{align*}
        f * \one = \one * f = f,
        \quad \text{for every series } f \in \series \Q \Sigma.
    \end{align*}
    By taking left derivatives and by the product rule,
    \begin{align*}
        \derive a (f * \one)
        = P(f, \derive a f, \one, \derive a \one)
        = P(f, \derive a f, \one, \eta \cdot \one)
        = \derive a f.
    \end{align*}
    We now extract the constant term of both sides, and since this operation is a homomorphism, we have
    \begin{align*}
        P(f_\e, f_a, 1, \eta) = f_a.
    \end{align*}
    Since $f$ is arbitrary we can choose $f_\e$ and $f_a$ as we wish.
    Since the base field of the rationals is infinite, this means
    \begin{align*}
        P(x, \dot x, 1, \eta)
            = \alpha \cdot x + \beta \cdot(\dot x + \eta \cdot x) + \gamma \eta \cdot \dot x = \dot x.
    \end{align*}
    We can now gather the coefficient of $x$ and of $\dot x$ and we obtain
    \begin{align*}
        (\alpha + \beta \eta) \cdot  x + (\beta - 1 + \gamma \eta) \cdot \dot x = 0.
    \end{align*}
    Equating to zero the coefficient of $x$ and that of $\dot x$,
    we obtain the following two constraints on $\eta$:
    \begin{align*}
        \alpha + \beta \eta &= 0
        \quad \text{and} \quad
        \beta - 1 + \gamma \eta = 0.
    \end{align*}
    Since the product is simple,
    $\alpha, \beta, \gamma$ satisfy~\cref{eq:simple},
    and since it is not degenerate,
    either $\beta \neq 0$ or $\gamma \neq 0$.
    If $\beta \neq 0$, then $\eta = - \alpha / \beta$ is a solution.
    If $\gamma \neq 0$, then $\eta = - (\beta - 1) / \gamma$ is a solution.
    In either case, we found a unit for the product.

    Assume now that the simple product is degenerate,
    so that $\beta = \gamma = 0$.
    We show that in this case the product has no unit at all
    (not even one not definable coinductively).
    Indeed, if $P = \alpha \cdot x y$,
    then for every series $f, g$ and input letter $a \in \Sigma$
    we have $\derive a (f * g) = \alpha \cdot f * g$.
    An induction on the length of words shows that
    the coefficient of $w \in \Sigma^*$ of the product is
    \begin{align*}
        \coefficient w {(f * g)} = \alpha^{\length w} \cdot f_\e \cdot g_\e.
    \end{align*}
    By way of contradiction, assume that there is a unit $\one$.
    Since $f * \one = f$ holds for every series $f$,
    by extracting the coefficient of $w \in \Sigma^*$
    we would have $\alpha^{\length w} \cdot f_\e = f_w$,
    which is impossible since $f_w$ can be chosen freely.
\end{proof}

\subsection{Zeroness algorithm}
\label{sec:zeroness algorithm}

We provide the missing details completing the proof of decidability of equivalence from \cref{sec:equivalence algorithm}.
We recall and prove the relevant properties of the polynomial ideals $I_n$.
\lemIdeals*
\begin{proof}
    The first point holds since the operation of constructing the ideal generated by a set
    is monotonic for inclusion of sets.
    
    The second point follows from the compatibility of the product rule $P$ with ideals~\cref{eq:ideal compatibility}.
    To see this, let $\alpha \in \Delta_a I_n$.
    Then $\alpha$ is of the form
    \begin{align*}
        &\alpha
        = \Delta_a \sum_{i = 1}^m \beta_i \cdot \Delta_{w_i} x_1 = 
            && \text{by~\cref{eq:extension}} \\
        & = \sum_{i = 1}^m P(\beta_i, \Delta_a \beta_i, \Delta_{w_i} x_1, \Delta_a \Delta_{w_i} x_1) =
            && \text{by~\cref{eq:ideal compatibility}} \\
        & = \sum_{i = 1}^m
            ( \mu_i \cdot \underbrace{\Delta_{w_i} x_1}_{\in I_n}
            + \eta_i \cdot \underbrace{\Delta_a \Delta_{w_i} x_1}_{\in I_{n+1}}),
    \end{align*}
    for some $\mu_i, \eta_i \in \poly \Q X$.
    Thus $\alpha$ is in $I_{n+1}$.
    
    We now consider the third point.
    The ``$\supseteq$'' inclusion holds by the first two points
    and the fact that $I_{n+1}$ is an ideal.
    For the other inclusion, let $\alpha \in I_{n+1}$.
    We can write $\alpha$ separating words of length exactly $n+1$ from the rest, and we have
    \begin{align*}
        &\alpha
        = \underbrace \beta_{\in I_n}
        + \sum_{w \in \Sigma^{n+1}} \beta_w \cdot \Delta_w x_1
        = \beta + \sum_{a \in \Sigma} \sum_{w \in \Sigma^n} \beta_{w \cdot a} \cdot \Delta_{w \cdot a} x_1 
        = \beta + \sum_{a \in \Sigma} \sum_{w \in \Sigma^n} \beta_{w \cdot a} \cdot \Delta_a \underbrace{\Delta_w x_1}_{\in I_n},
    \end{align*}
    where the latter quantity is clearly in the ideal generated by $\Delta_a I_n$.
    
    The last point follows from the previous one
    since we have shown that $I_{n+1}$ is (exclusively) a function of $I_n$.
\end{proof}

We recall and prove the following zeroness criterion.

\lemZeronessCharacterisation*
\begin{proof}
    The ``only if'' direction is trivial.
    For the other direction, assume that the semantics is zero for words of length $\leq N$.
    Thanks to the homomorphism property of the semantics (\cref{lem:polynomial automata:homomorphism}),
    for every word $w \in \Sigma^*$ we can write
    \begin{align*}
        \deriveleft w \sem {x_1}
        &= \sem {\Delta_w x_1}
        = \sem {\sum_{u \in \Sigma^{\leq N}} \beta_u \cdot \Delta_u x_1}
        = \sum_{u \in \Sigma^{\leq N}} \sem {\beta_u} * \sem {\Delta_u x_1},
    \end{align*}
    where we have used the fact that $\Delta_w x_1$ is in $I_N$.
    By taking constant terms on both sides,
    \begin{align*}
        &\coefficient w \sem {x_1}
        = \coefficient \e (\deriveleft w \sem {x_1})
        = \sum_{u \in \Sigma^{\leq N}} \coefficient \e \sem {\beta_u} \cdot \coefficient \e \sem {\Delta_u x_1}
        = \sum_{u \in \Sigma^{\leq N}} \coefficient \e \sem {\beta_u} \cdot \coefficient u \sem {x_1}.
    \end{align*}
    But $u$ has length $\leq N$, and thus by the assumption $\coefficient u \sem {x_1} = 0$,
    implying $\coefficient w \sem {x_1} = 0$. 
\end{proof}

By combining the results above,
we can finally prove decidability of equivalence of finite-variable $P$-automata
when $P$ is a special product rule.
\thmDecidability*
\begin{proof}
    By using the effective closure properties of $P$-finite series,
    we reduce equality $f = g$ to zeroness $f - g = \zero$.
    So let $f = \Asem \AA {x_1}$ be a $P$-finite series,
    recognised by a $P$-automaton $\AA$.
    Compute the stabilisation index $N \in \N$ of the ideal chain~\cref{eq:ideal chain},
    which is possible by~\cref{lem:ideals} and decidability of the ideal membership problem.
    Thanks to~\cref{lem:zeroness characterisation},
    we can decide $f = \zero$
    by enumerating all words $w$ of length $\leq N$
    and checking $\coefficient w f = 0$.
    The latter test is performed by applying the semantics of the $P$-automaton.
\end{proof}


\agdasection{Right derivatives, reversal, and the commutativity problem}{General/Reversal}
\label{app:commutativity}

\renewcommand{\deriveleft}[1]{\delta^L_{#1}}

\begin{table*}[h!]
    \begin{center}
        \makegapedcells
        \begin{tabular}{l|l|l|l}
            operation & notation & initial value & left derivative by $a$ \\
            \hline\hline
            right derivative by $b \in \Sigma$
                & $\deriveright b f$
                & $\coefficient \e (\deriveright b f) = \coefficient b f$
                & $\deriveleft a (\deriveright b f) = \deriveright b (\deriveleft a f)$ \\
            reversal
                & $\reverse f$
                & $\coefficient \e (\reverse f) = \coefficient \e f$
                & $\deriveleft a (\reverse f) = \reverse {(\deriveright a f)}$ \\
        \end{tabular}
    \end{center}
    \caption{Right derivatives and reversal operations, defined coinductively.}
    \label{table:other operations on series}
\end{table*}

In this section, we show that the commutativity problem is decidable for $P$-finite series and $P$-automata,
when $P$ is a special product rule.
This is a nontrivial development, extending equivalence to a more general decision problem.

Our approach is as follows.
By the results in~\cite[Sec.~II]{Clemente:LICS:2025},
in order to decide commutativity,
it suffices to show that $P$-finite series are effectively closed under \emph{right derivatives}~(\cref{lem:commutativity characterisation}).
In order to show the latter, it suffices to show that
right derivatives satisfy \emph{any} product rule~(\cref{lem:closure under right derivatives}).
In fact, we will show a stronger result, namely that, for special $P$-products,
right derivatives satisfy \emph{the same} product rule as left derivatives do.
To obtain the latter, we use the fact that the reversal operation
is an \emph{endomorphism} of the series algebra~(\cref{lem:reversal endomorphism}).
Note that even when $P$-finite series are closed under right derivatives,
they need not be closed under the reversal operation itself.

In order to execute this plan, in~\cref{lem:reversal endomorphism:characterisation} we provide an equational characterisation
of when reversal is an endomorphism for an arbitrary $P$-product.
This equational characterisation is similar in spirit to~\cref{thm:characterisation:BAC} characterising \BAC~product rules,
but more advanced.
Based on this characterisation, we can directly check that it holds
for simple (thus special) product rules~(\cref{lem:reversal endomorphism:special}).

The rest of the section is organised as follows.
We begin in~\cref{sec:commutativity:preliminaries} with some technical preliminaries.
In~\cref{sec:commutativity:general results} we show that, when reversal is an endomorphism of the series algebra,
the class of $P$-finite series is effectively closed under right derivatives;
this result holds for every product rule $P$.
In~\cref{sec:commutativity:special results} we show that, for special products rules,
reversal is always an endomorphism,
and consequently $P$-finite series are closed under right derivatives.
This allows us to conclude that commutativity is decidable in~\cref{sec:commutativity:decidability}.
In fact, we show that commutativity and equivalence are inter-reducible.

\subsection{Preliminaries}
\label{sec:commutativity:preliminaries}

\subparagraph*{Right derivatives and reversal.}
The \emph{right derivative} of a series $f \in \series \Q \Sigma$ with respect to a letter $b \in \Sigma$
is the series $\deriveright b f$
that maps an input word $w \in \Sigma^*$ to $f(w \cdot b)$.
In this section, to avoid confusion we denote left derivatives by $\deriveleft a$
and right derivatives by $\deriveright b$.
In line with the style of this paper,
it is also possible to define right derivatives coinductively, based on left derivatives;
see~\cref{table:other operations on series}.
Like left derivatives, their right counterparts are extended to all finite words homomorphically:
$\deriveright \e f := f$ and $\deriveright {a \cdot w} f := \deriveright w (\deriveright a f)$
for all $a \in \Sigma$ and $w \in \Sigma^*$.

The \emph{reversal} operation $\reverse f$
can be defined in a point-wise manner by
$\coefficient w {\reverse f} := \coefficient {\reverse w} f$
where $\reverse w$ is the mirror image of the word $w$.
The coinductive definition based on left and right derivatives is given in~\cref{table:other operations on series}.

\subparagraph*{Bisimulations up-to reflexivity-transitivity.}

Sometimes we will need a stronger proof principle for series equality.
For a binary relation on series $R \subseteq \series \Q \Sigma \times \series \Q \Sigma$,
let $R^*$ be its reflexive-transitive closure.
%
%
A \emph{bisimulation up-to reflexivity-transitivity} is a relation $R$ on series
\st~whenever $f \mathrel R g$ holds, we have
\begin{enumerate}[label=\textbf{(C.\arabic*)}]

    \item \label{bisimulation upto:1}
    $f_\e = g_\e$, and

    \item \label{bisimulation upto:2}
    for every letter $a \in \Sigma$,
    we have $(\deriveleft a f) \mathrel {R^*} (\deriveleft a g)$.

\end{enumerate}
We have the following equality proof principle for bisimulations up-to reflexivity-transitivity.
Its proof uses the fact that if $R$ is a bisimulation up-to reflexivity-transitivity,
then $R^*$ is included in a bisimulation (since the largest simulation \emph{is} reflexive and transitive),
and then one applies~\cref{lem:bisimulation}.

\begin{lemma}
    \label{lem:bisimulation up-to reflexivity-transitivity}
    If $f \mathrel R g$ for some bisimulation up-to reflexivity-transitivity $R$,
    then $f = g$.
\end{lemma}

\subsection{General results}
\label{sec:commutativity:general results}

The results in this section hold for every product rule $P$.
We are interested in an additional closure property for $P$-finite series not covered by~\cref{lem:P-finite:closure properties},
namely closure under right derivatives $\deriveright b$.
Since $P$-finiteness is defined in terms of left derivatives $\deriveleft a$,
this property does not follow directly from the definitions.
What is missing is a product rule for right derivatives
satisfying a coinductive equation similar to~\cref{eq:product rule}.
We show that when the reversal operation $\reverse{(\_)}$ is an endomorphism,
then right derivatives satisfy precisely \emph{the same product} rule as left derivatives.
In fact these two conditions are equivalent.

\subsubsection{Reversal and right derivatives}

We say that the reversal operation is an \emph{endomorphism} if,
for all $f, g \in \series \Q \Sigma$, we have
\begin{equation}
    \label{eq:reversal endomorphism}
    \tag{rev-end}
    \left\{\ 
        \begin{aligned}
            \reverse{(c \cdot f)}
                &= c \cdot \reverse f, \text{ for all } c \in \Q, \\
            \reverse {(f + g)}
                &= \reverse f + \reverse g, \\
            \reverse{(f * g)}
                &= \reverse f * \reverse g.
        \end{aligned}
    \right.
\end{equation}
Note that the first two equations, enforcing linearity, hold for every product rule $P$.
The real restriction is the last equation, requiring that reversal commutes with the $P$-product.

\subsubsection{Semantic characterisation}

In the following lemma we characterise when this happens in terms of a product rule satisfied by right derivatives.
\begin{agdalemma}[General/Reversal/\#sec:rev-product_rule-characterisation]
    \label{lem:reversal endomorphism}
    Consider a $P$-product ``$*$'' defined by a product rule $P$
    satisfying~\cref{eq:product rule}.
    The reversal operation is an endomorphism~\cref{eq:reversal endomorphism}
    if, and only if, for every $f, g \in \series \Q \Sigma$ and $b \in \Sigma$, we have
    \begin{align}
        \label{eq:product rule:right derivative}
        \tag{$*$-$\deriveright {}$}
        \deriveright b (f * g)
            &= 
            P(f, \deriveright b f, g, \deriveright b g).
    \end{align}
\end{agdalemma}
\begin{proofsketch}
    For the ``only if'' direction,
    assume that reversal is an endomorphism.
    Then~\cref{eq:product rule:right derivative} is proved by
    expressing $\deriveright b (f * g)$ by double reversal as $\reverse{(\deriveleft b \reverse{(f * g)})}$
    and then applying the product rule for left derivatives~\cref{eq:product rule}.
    
    For the ``if'' direction, assume that~\cref{eq:product rule:right derivative} holds.
    The first two cases of~\cref{eq:reversal endomorphism} hold by the definition of reversal.
    The last case (product) is shown as follows.
    Let ``$\sim$'' be the relation on series defined by the following rules
    %
    \begin{align*}
        \reverse {(c \cdot f)}
            \sim c \cdot \reverse f,
        \quad
        \reverse {(f + g)}
            \sim \reverse f + \reverse g,
        \quad \text{ and } \quad
        \reverse {(f * g)}
            \sim \reverse f * \reverse g.
    \end{align*}
    We show that ``$\sim$'' is a bisimulation up-to reflexivity-transitivity,
    to which we apply~\cref{lem:bisimulation up-to reflexivity-transitivity}.
    Constant terms agree~\cref{bisimulation upto:1},
    and regarding left derivatives~\cref{bisimulation upto:2} we have
    \begin{align*}
        &\deriveleft a \reverse {(f * g)}
        = \reverse {(\deriveright a (f * g))}
            && \text{(def.~of reverse)} \\
        &\ = \reverse {(P(f, \deriveright a f, g, \deriveright a g))} 
            && \text{(by~\cref{eq:product rule:right derivative})} \\
        &\ \mathrel {\sim^*} P(\reverse f, \reverse {(\deriveright a f)}, \reverse g, \reverse {(\deriveright a g)})
            && \text{(def.~of ``$\sim^*$'')} \\
        &\ = P(\reverse f, \deriveleft a \reverse f, \reverse g, \deriveleft a \reverse g)
            && \text{(def.~of reverse)} \\
        &\ = \deriveleft a (\reverse f * \reverse g).
            && \text{(by~\cref{eq:product rule})}
        \qedhere
    \end{align*}
\end{proofsketch}

In the following lemma we show that when right derivatives satisfy \emph{any} product rule
(not necessarily the same rule as for left derivatives),
then $P$-finite series are closed under right derivatives.
\begin{agdalemma}[General/Reversal/\#sec:right-derivatives-P-fin]
    \label{lem:closure under right derivatives}
    Suppose that the right derivatives satisfy a product rule of the form
    \begin{align*}
        \deriveright b (f * g) = Q(f, \deriveright b f, g, \deriveright b g),
    \end{align*}
    for an \emph{arbitrary} polynomial $Q(x, \dot x, y, \dot y)$.
    If $f \in \series \Q \Sigma$ is $P$-finite,
    then for every $b \in \Sigma$, the series $\deriveright b f$ is also $P$-finite.
\end{agdalemma}
\begin{proofsketch}
    If $f$ is $P$-finite with generators $g_1, \dots, g_k$,
    then we claim that $\deriveright b f$ is $P$-finite with generators
    $g_1, \dots, g_k$, $\deriveright b g_1, \dots, \deriveright b g_k$.
    It suffices to notice that for every input letter $a \in \Sigma$ and new generator $\deriveright b g_i$,
    $\deriveleft a (\deriveright b g_i)$ belongs to the algebra generated by all generators (new and old).
    Indeed, right and left derivatives commute
    $\deriveleft a (\deriveright b g_i) = \deriveright b (\deriveleft a g_i)$,
    so that $\deriveleft a (\deriveright b g_i)$
    is of the form $\deriveright b (R(g_1, \dots, g_k))$
    for some term $R$ depending on $a$ and $i$.
    Since right derivatives satisfy a product rule,
    and induction on the structure of $R$ shows that the latter quantity
    is of the form $S(g_1, \dots, g_k, \deriveright b g_1, \dots, \deriveright b g_k)$.
\end{proofsketch}

By combining~\cref{lem:reversal endomorphism} and~\cref{lem:closure under right derivatives},
we obtain closure under right derivatives.

\begin{agdacorollary}[General/Reversal/\#sec:rev-end-right-derivatives-P-fin]
    \label{cor:closure under right derivatives}
    If the reversal operation is an endomorphism,
    then $P$-finite series are effectively closed under right derivatives.
\end{agdacorollary}

\agdasubsubsection{Automaton-based characterisation}{General/ReversalEnd/\#sec:automata}

We extend the characterisation from~\cref{lem:reversal endomorphism}
by adding two new equivalent conditions based on a $P$-automaton (over infinitely many variables).
This will be used later to show that reversal is an endomorphism for all \BAC~$P$-products.

We define a $P$-automaton $\AA$
that generalises the syntactic $P$-automaton from~\cref{rem:product via automata}.
Consider a set of variables $X$ of the form $\lxr u x f v$
where $f \in \series \Q \Sigma$ is any series and $u, v \in \Sigma^*$.
Intuitively, this represents the series $f$ after having read $u$ from the left and $v$ from the right.
For convenience, we just write $x_f$ for the variable $\lxr \e x f \e$.
Consider the two transition functions $\Delta^L_a, \Delta^R_b : X \to \Terms X$
\st~for all $a, b \in \Sigma$,
\begin{align}
    \label{eq:left-right transitions}
    \tag{$\Delta^L_a, \Delta^R_b$}
    \Delta^L_a (\lxr u x f v)
        := \lxr {u \cdot a} x f v
    \quad\text{and}\quad
    \Delta^R_b (\lxr u x f v)
        := \lxr u x f {v \cdot b},
\end{align}
and extend both of them to $\Terms X$ by~\cref{eq:extension}.
Consider the $P$-automaton $\AA = \tuple{\Sigma, X, F, \Delta^L}$
with the left transition function above,
and output function given by $F(\lxr u x f v) := \coefficient {u \cdot \reverse v} f$.
The construction is set up so that $\Asem \AA {x_f} = f$.
More generally, the following observation holds.
\begin{agdaproposition}[General/ReversalEnd/\#var-lemma]
    \label{prop:variable case}
    $\Asem \AA {\lxr u x f v} = \deriveleft u \deriveright v f$.
\end{agdaproposition}

Thanks to this automaton, we can extend the \emph{semantic} characterisation from~\cref{lem:reversal endomorphism}
with two additional, equivalent, \emph{automata} conditions.
This holds for every product rule $P$.
It will provide the bridge with the forthcoming \emph{syntactic} characterisation in~\cref{lem:reversal endomorphism:characterisation},
which will hold only for special product rules $P$.
\begin{lemma}
    \label{lem:reversal endomorphism:automata}
    The reversal operation is an endomorphism
    if, and only if, any of the following equivalent conditions hold:
    \begin{align}
        \label{eq:first}
        \tag{$\deriveright {}$-$\Delta^R$}
        \deriveright b (\Asem \AA \alpha)
            &= \Asem \AA {\Delta^R_b \alpha}
            &&\forall \alpha, b \\
        \label{eq:second}
        \tag{$\sem {\Delta^R \text- \Delta^L}$}
        \Asem \AA {\Delta^R_b \Delta^L_a \alpha}
            &= \Asem \AA {\Delta^L_a \Delta^R_b \alpha}
            &&\forall a, b, \alpha
    \end{align}
\end{lemma}
\begin{proof}
    \cref{eq:product rule:right derivative}$\implies$\cref{eq:first}:
    The proof proceeds by induction on the structure of $\alpha$.
    The base case $\alpha = \lxr u x f v$ holds by~\cref{prop:variable case}:
    \begin{align*}
        \deriveright b (\Asem \AA {\lxr u x f v})
        &= \deriveright b \deriveleft u \deriveright v f
        = \deriveleft u \deriveright {v \cdot b} f
        = \Asem \AA {\Delta^R_b (\lxr u x f v)}.
    \end{align*}
    The inductive cases for addition and scalar multiplication follow by linearity.
    For the product case we use~\cref{lem:homomorphism} and~\cref{eq:product rule:right derivative}:
    \begin{align*}
        &\deriveright b (\Asem \AA {\alpha * \beta})
        = \deriveright b (\Asem \AA \alpha * \Asem \AA \beta)
            && \text{by~\cref{lem:homomorphism}} \\
        &\ = P(\Asem \AA \alpha, \deriveright b (\Asem \AA \alpha), \Asem \AA \beta, \deriveright b (\Asem \AA \beta))
            && \text{by~\cref{eq:product rule:right derivative}} \\
        &\ = P(\Asem \AA \alpha, \Asem \AA {\Delta^R_b \alpha}, \Asem \AA \beta, \Asem \AA {\Delta^R_b \beta})
            && \text{by ind.} \\
        &\ = \Asem \AA {P(\alpha, \Delta^R_b \alpha, \beta, \Delta^R_b \beta)}
            && \text{by~\cref{lem:homomorphism}} \\
        &\ = \Asem \AA {\Delta^R_b (\alpha * \beta)}.
            && \text{by~\cref{eq:extension}}
    \end{align*}
    \cref{eq:first}$\implies$\cref{eq:product rule:right derivative}:
    \begin{align*}
        &\deriveright b (f * g)
        = \deriveright b (\Asem \AA {x_f} * \Asem \AA{x_g}) =
            && \text{by~\cref{prop:variable case}} \\
        &\ = \deriveright b (\Asem \AA {x_f * x_g})
            && \text{by~\cref{lem:homomorphism}} \\
        &\ = \Asem \AA {\Delta^R_b (x_f * x_g)}
            && \text{by~\cref{eq:first}} \\
        &\ = \Asem \AA {P(x_f, \Delta^R_b x_f, x_g, \Delta^R_b x_g)}
            && \text{by~\cref{eq:extension}} \\
        &\ = P(\Asem \AA {x_f}, \Asem \AA {\Delta^R_b x_f}, \Asem \AA {x_g}, \Asem \AA {\Delta^R_b x_g})
            && \text{by~\cref{lem:homomorphism}} \\
        &\ = P(\Asem \AA {x_f}, \deriveright b (\Asem \AA {x_f}), \Asem \AA {x_g}, \deriveright b (\Asem \AA {x_g}))
            && \text{by~\cref{eq:first}} \\
        &\ = P(f, \deriveright b f, g, \deriveright b g).
            && \text{by~\cref{prop:variable case}}
    \end{align*}
    \cref{eq:first}$\implies$\cref{eq:second}:
    \begin{align*}
        &\Asem \AA {\Delta^R_b \Delta^L_a \alpha}
        = \deriveright b (\Asem \AA {\Delta^L_a \alpha})
            && \text{by~\cref{eq:first}} \\
        &\ = \deriveright b (\deriveleft a (\Asem \AA \alpha))
            && \text{by~\cref{eq:automata semantics}} \\
        &\ = \deriveleft a (\deriveright b (\Asem \AA \alpha))
            && \text{by def.} \\
        &\ = \deriveleft a (\Asem \AA {\Delta^R_b \alpha})
            && \text{by~\cref{eq:first}} \\
        &\ = \Asem \AA {\Delta^L_a \Delta^R_b \alpha}.
            && \text{by~\cref{eq:automata semantics}}
    \end{align*}
    \cref{eq:second}$\implies$\cref{eq:first}:
    We show that the relation on series
    \begin{align*}
        \deriveright b (\Asem \AA \alpha) \sim \Asem \AA {\Delta^R_b \alpha}
    \end{align*}
    is a bisimulation.
    \begin{itemize}
        \item The base case amounts to showing
        \begin{align*}
            \coefficient \e (\deriveright b (\Asem \AA \alpha))
            = \coefficient \e (\Asem \AA {\Delta^R_b \alpha}).
        \end{align*}
        Since
        %
            $\coefficient \e (\deriveright b (\Asem \AA \alpha))
            = \coefficient b (\Asem \AA \alpha)
            = \coefficient \e (\deriveleft b (\Asem \AA \alpha))$,
        %
        by~\cref{eq:automata semantics} it suffices to show the following equality of rational numbers
        \begin{align*}
            F ({\Delta^L_b \alpha})
            = F ({\Delta^R_b \alpha}),
        \end{align*}
        which can be done by structural induction on $\alpha$.
        Notice that $\Delta^L_b \alpha$ is different from $\Delta^R_b \alpha$ in general,
        however we show that their initial values coincide.
        The variable case $\alpha = \lxr u x f v$ follows by the definitions of $\Delta^R_b$ and $F$.
        The cases of scalar multiplication and addition follow by linearity.
        Finally, for the product case since $F$ is a homomorphism, we have
        \begin{align*}
            &F (\Delta^L_b (\alpha * \beta))
            = F (P(\alpha, \Delta^L_b \alpha, \beta, \Delta^L_b \beta)) = 
                &&\text{(by~\cref{eq:extension})} \\
            &= P(F \alpha, F (\Delta^L_b \alpha), F \beta, F (\Delta^L_b \beta)) =
                &&\text{($F$ hom.)} \\
            &= P(F \alpha, F (\Delta^R_b \alpha), F \beta, F (\Delta^R_b \beta)) =
                &&\text{(by ind.)} \\
            &= F (P(\alpha, \Delta^R_b \alpha, \beta, \Delta^R_b \beta)) =
                &&\text{($F$ hom.)} \\
            &= F (\Delta^R_b (\alpha * \beta)).
                &&\text{(by~\cref{eq:extension})}            
        \end{align*}

        \item For the coinductive case, we have
        \begin{align*}
            \deriveleft a \deriveright b (\Asem \AA \alpha)
            &= \deriveright b \deriveleft a (\Asem \AA \alpha) = 
                &&\text{by def.} \\
            &= \deriveright b (\Asem \AA {\Delta^L_a \alpha})
                &&\text{by~\cref{eq:automata semantics}} \\
            &\sim \Asem \AA {\Delta^R_b \Delta^L_a \alpha} =
                &&\text{by def.} \\
            &= \Asem \AA {\Delta^L_a \Delta^R_b \alpha} =
                &&\text{by~\cref{eq:second}} \\
            &= \deriveleft a (\Asem \AA {\Delta^R_b \alpha})
                &&\text{by~\cref{eq:automata semantics}}.
            \qedhere
        \end{align*}
    \end{itemize}
    %
\end{proof}

\subsection{Special results}
\label{sec:commutativity:special results}

The results in this section hold for special product rules $P$.
We prove that reversal is an endomorphism of the series algebra~\cref{eq:reversal endomorphism}
for \emph{every} $P$-product arising from a special product rule $P$.
In particular, reversal is an endomorphism
for the Hadamard, shuffle, and infiltration algebras
(subsuming~\cite[Lemmas 4 and 10]{Clemente:LICS:2025} and \cite[Lemma 19]{Clemente:arXiv:LICS:2025}),
and thus, by~\cref{lem:closure under right derivatives},
the corresponding classes of Hadamard, shuffle, and infiltration-finite series are closed under right derivatives
(subsuming~\cite[Lemmas 8 and 12]{Clemente:LICS:2025} and \cite[Lemma 20]{Clemente:arXiv:LICS:2025}).
In order to prove this result,
in the next section we extend the semantic characterisation from~\cref{lem:reversal endomorphism}
and the automaton-based characterisation from~\cref{lem:reversal endomorphism:automata}
to an \emph{equational} characterisation.

\subsubsection{Equational characterisation}

%
%
\begin{agdalemma}[Special/Reversal]
    \label{lem:reversal endomorphism:characterisation}
    Let $P$ be a special product rule.
    The reversal operation is an endomorphism
    if, and only if, any of the following two equations hold:
    \begin{align}
        \label{eq:third}
        \tag{$\Delta^R$-$\Delta^L$}
        \Delta^R_b \Delta^L_a \alpha
            &\approx \Delta^L_a \Delta^R_b \alpha,
                &&\forall \alpha, a, b, \\
        \label{eq:fourth}
        \tag{$\Delta^R$-$\Delta^L$-$x$}
        \Delta^R_b \Delta^L_a (x_f * x_g)
            &\approx \Delta^L_a \Delta^R_b (x_f * x_g),
                &&\forall a, b, f, g,
    \end{align}
    where $a, b \in \Sigma$, $\alpha \in \Terms X$,
    and $f, g \in \series \Q \Sigma$.
\end{agdalemma}
\noindent
Notice that the second equation~\cref{eq:fourth}, even if it quantifies over series $f, g$,
in fact it only involves two \emph{distinct} variables $x_f, x_g$,
and thus it is a decidable condition.
\begin{proof}
    Thanks to the characterisation from~\cref{lem:reversal endomorphism:automata},
    it suffices to show equivalence with~\cref{eq:second}.
    We show the following three implications
    \cref{eq:third}$\implies$\cref{eq:second}$\implies$\cref{eq:fourth}$\implies$\cref{eq:third}.
    
    \cref{eq:third}\agda{$\implies$}{Special/Reversal/\#sec:proof:1}\cref{eq:second}:
    This direction follows directly by the invariance of the semantics.

    \cref{eq:second}$\implies$\cref{eq:fourth}:
    Consider arbitrary series $f, g$ and instantiate~\cref{eq:second} with $\alpha := x_f * x_g$, obtaining
    \begin{align*}
        \Asem \AA {\Delta^R_b \Delta^L_a (x_f * x_g)}
        &= \Asem \AA {\Delta^L_a \Delta^R_b (x_f * x_g)}
        &&\forall a, b \in \Sigma.
    \end{align*}
    The terms inside the semantic brackets are of the form
    \begin{align*}
        \Delta^R_b \Delta^L_a (x_f * x_g)
        &= P_{ab}
            (x_f, \lxr a x f {}, \lxr {} x f b, \lxr a x f b,
            x_g, \lxr a x g {}, \lxr {} x g b, \lxr a x g b) \text{, resp., } \\
        \Delta^L_a \Delta^R_b (x_f * x_g)
        &= P_{ba}
            (x_f, \lxr a x f {}, \lxr {} x f b, \lxr a x f b,
            x_g , \lxr a x g {}, \lxr {} x g b, \lxr a x g b),
    \end{align*}
    for some $P_{ab}, P_{ba} \in \Terms X$.
    Thus, by the homomorphism property of the semantics~(\cref{lem:homomorphism})
    and~\cref{prop:variable case} we have
    \begin{align*}
        P_{ab}
            (f, \deriveleft a f, \deriveright b f, \deriveleft a \deriveright b f,
            g, \deriveleft a g, \deriveright b g, \deriveleft a \deriveright b g)
        = P_{ba}
            (f, \deriveleft a f, \deriveright b f, \deriveleft a \deriveright b f,
            g, \deriveleft a g, \deriveright b g, \deriveleft a \deriveright b g).
    \end{align*}
    By taking the constant term of both sides above,
    by~\cref{lem:constant term homomorphism} we obtain
    \begin{align*}
        P_{ab}(f_\e, f_a, f_b, f_{ab}, g_\e, g_a, g_b, g_{ab})
        = P_{ba}(f_\e, f_a, f_b, f_{ab}, g_\e, g_a, g_b, g_{ab}).
    \end{align*}
    But $f, g$ are arbitrary, and thus the two polynomial functions represented by $P_{ab}, P_{ba}$ are equal.
    Over an infinite field (such as $\Q$),
    they are equal as polynomials~\cite[Corollary 6]{CoxLittleOShea:Ideals:2015} and we obtain~\cref{eq:fourth}.
    
    \cref{eq:fourth}\agda{$\implies$}{Special/Reversal/\#sec:proof:2}\cref{eq:third}:
    We proceed by induction on the structure of $\alpha$.
    Regarding the base case of a variable $\alpha = \lxr u x f v$,
    we have that $\Delta^R_b, \Delta^L_a$ commute by definition:
    \begin{align*}
        \Delta^R_b \Delta^L_a (\lxr u x f v)
            = \Delta^R_b (\lxr {u \cdot a} x f v)
            = \lxr {u \cdot a} x f {v \cdot b}
        \quad \text{ and } \quad
        \Delta^L_a \Delta^R_b (\lxr u x f v)
            = \Delta^L_a (\lxr u x f {v \cdot b})
            = \lxr {u \cdot a} x f {v \cdot b}.
    \end{align*}
    The cases of addition and scalar product follow by linearity and the inductive assumption.
    We finally consider the case of a product $\alpha * \beta$.
    Pick any two distinct series $f, g \in \series \Q \Sigma$
    and consider the distinct variables $y := x_f$ and $z := x_g$.
    Let $Y$ be the finite set of variables
    \begin{align*}
        Y := \set{
            y, \lr a y {}, \lr {} y b, \lr a y b,
            z, \lr a z {}, \lr {} z b, \lr a z b}
    \end{align*}
    and let $\rho : Y \to \Terms X$ be the variable substitution defined as follows:
    \begin{align*}
        \rho y &:= \alpha,
            &\rho z &:= \beta, \\
        \rho (\lr a y {}) &:= \Delta^L_a \alpha,
            &\rho (\lr a z {}) &:= \Delta^L_a \beta, \\
        \rho (\lr {} y b) &:= \Delta^R_b \alpha,
            &\rho (\lr {} z b) &:= \Delta^R_b \beta, \\
        \rho (\lr a y b) &:= \Delta^R_b \Delta^L_a \alpha \approx \Delta^L_a \Delta^R_b \alpha,
            &\rho (\lr a z b) &:= \Delta^R_b \Delta^L_a \beta \approx \Delta^L_a \Delta^R_b \beta,
    \end{align*}
    where the two equivalences at the bottom hold by the inductive assumption.
    \begin{claim*}
        We have the following two commuting properties:
        \begin{align}
            \tag{$\Delta^L_a$-$\rho$}
            \label{eq:commute-rule-1}
            \Delta^L_a \rho \gamma
                &\approx \rho \Delta^L_a \gamma,
                &&\text{ for all } \gamma \in \Terms {\set {y, \lr {} y b, z, \lr {} z b}}, \\
            \tag{$\Delta^R_b$-$\rho$}
            \label{eq:commute-rule-2}
            \Delta^R_b \rho \gamma
                &\approx \rho \Delta^R_b \gamma,
                &&\text{ for all } \gamma \in \Terms {\set {y, \lr a y {}, z, \lr a z {}}}.
        \end{align}
    \end{claim*}
    \noindent
    We note that in~\cref{eq:commute-rule-1,eq:commute-rule-2},
    the terms $\Delta^L_a \gamma, \Delta^R_b \gamma$ belong to $\Terms Y$,
    and thus $\rho$ can be applied to it.
    \begin{claimproof}
        We prove~\cref{eq:commute-rule-1} by structural induction on $\gamma$.
        The base cases $\gamma = y$ and $\gamma = \lr {} y b$ hold by definition:
        \begin{align*}
            \Delta^L_a \rho y &= \Delta^L_a \alpha
                &
                    \rho \Delta^L_a y
                    &= \rho (\lr a y {}) = \Delta^L_a \alpha,
                    \text{ resp.,} \\
            \Delta^L_a \rho (\lr {} y b) &= \Delta^L_a \Delta^R_b \alpha
                &
                    \rho \Delta^L_a \lr {} y b
                    &= \rho (\lr a y b) = \Delta^L_a \Delta^R_b \alpha.
        \end{align*}
        The base cases $\gamma = z$ and $\gamma = \lr {} z b$ are analogous.
        The inductive cases $\gamma = \gamma_0 + \gamma_1$ and $\gamma = c \cdot \gamma_0$ follow by linearity.
        Regarding the inductive case $\gamma = \gamma_0 * \gamma_1$, we have
        \begin{align*}
            &\Delta^L_a \rho (\gamma_0 * \gamma_1)
            = \Delta^L_a (\rho \gamma_0 * \rho \gamma_1) = 
                && \text{(def.~of $\rho$)} \\
            &\ = P_a(\rho \gamma_0, \Delta^L_a \rho \gamma_0, \rho \gamma_1, \Delta^L_a \rho \gamma_1) =
                && \text{(by~\cref{eq:extension})} \\
            &\ \approx P_a(\rho \gamma_0, \rho \Delta^L_a \gamma_0, \rho \gamma_1, \rho \Delta^L_a \gamma_1) =
                && \text{(by ind.)} \\
            &\ = \rho P_a(\gamma_0, \Delta^L_a \gamma_0, \gamma_1, \Delta^L_a \gamma_1) =
                && \text{(def.~of $\rho$)} \\
            &\ = \rho \Delta^L_a (\gamma_0 * \gamma_1).
                && \text{(by~\cref{eq:extension})}
    \end{align*}
        The proof for~\cref{eq:commute-rule-2} is analogous.
    \end{claimproof}
    Thanks to the claim, we have
    \begin{align*}
        &\Delta^R_b \Delta^L_a (\alpha * \beta)
        = \Delta^R_b \Delta^L_a \rho (y * z)
            && \text{(def.~of $\rho$)} \\
        &\ \approx \Delta^R_b \rho \Delta^L_a (y * z)
            && \text{(by~\cref{eq:commute-rule-1})} \\
        &\ \approx \rho \Delta^R_b \Delta^L_a (y * z)
            && \text{(by~\cref{eq:commute-rule-2})} \\
        &\ \approx \rho \Delta^L_a \Delta^R_b (y * z)
            && \text{(by~\cref{eq:fourth})} \\
        &\ \approx \Delta^L_a \rho \Delta^R_b (y * z)
            && \text{(by~\cref{eq:commute-rule-1})} \\
        &\ \approx \Delta^L_a \Delta^R_b \rho (y * z)
            && \text{(by~\cref{eq:commute-rule-2})} \\
        &\ = \Delta^L_a \Delta^R_b (\alpha * \beta)
            && \text{(def.~of $\rho$)},
    \end{align*}
    as required.
    Notice that the two applications of~\cref{eq:commute-rule-1} are correct
    since $y * z$, resp., $\Delta^R_b (y * z)$ are in $\Terms {\set{y, \lr {} y b, z, \lr {} z b}}$.
    An analogous consideration applies to~\cref{eq:commute-rule-2}.
\end{proof}

%
%
As a consequence of~\cref{lem:reversal endomorphism:characterisation},
we have that reversal is an endomorphism for all special products.
\begin{lemma}
    \label{lem:reversal endomorphism:special}
    Reversal is an endomorphism of~\cref{eq:series algebra}
    for every $P$-product defined by a special product rule $P$.
\end{lemma}
\begin{proof}
    Thanks to~\cref{lem:reversal endomorphism:characterisation}
    and the fact that special product rules are simple (\cref{cor:classification}),
    it suffices to verify~\cref{eq:fourth} on simple product rules~\cref{eq:simple}.
    So let $P(x, \dot x, y, \dot y)$ be a simple product rule, thus of the form~\cref{eq:simple}.
    Write $x := x_f$ and $y := x_g$ for two distinct arbitrary series $f, g \in \series \Q \Sigma$.
    We need to show $\Delta_b^R \Delta_a^L (x * y) \approx \Delta_a^L \Delta_b^R (x * y)$.
    This rather tedious calculation can be done by hand (it involves 16 terms on each side),
    or more easily by a computer algebra system.
    %
\end{proof}

Combining~\cref{cor:closure under right derivatives} and~\cref{lem:reversal endomorphism:special},
we deduce the following closure property.

\begin{corollary}
    \label{cor:closure under right derivatives:special}
    Let $P$ be a special product rule.
    The class of $P$-finite series is effectively closed under right derivatives.
\end{corollary}
In the next section we leverage on this closure property
to show that the commutativity problem is decidable for $P$-finite series when $P$ is a special product rule.

\subsection{Commutativity problem}
\label{sec:commutativity:decidability}

A series $f \in \series \Q \Sigma$ is \emph{commutative} if
for every two words $u, v \in \Sigma^*$ with the same commutative image,
we have $\coefficient u f = \coefficient v f$.
The \emph{commutativity problem} for $P$-finite series
asks whether a given $P$-finite series is commutative.

The commutativity problem has applications in control theory~\cite[Proposition II.9]{Fliess:1981},
where it can be used to characterise the Chen-Fliess series
giving rise to analytic functionals~\cite[Theorem 10]{Clemente:arXiv:LICS:2025},
and in the theory of multivariate differential and difference equations,
where it can be used to decide solvability in power series (in commuting variables)
of nontrivial classes of differential and difference equations~\cite[Problems 1 and 2]{Clemente:LICS:2025}.

\subparagraph*{Decidability of the commutativity problem.}

We have the following elementary finite equational characterisation of commutativity.
This is where right derivatives come into play.
\begin{lemma}[\protect{\cite[Lemma~2]{Clemente:LICS:2025}}]
    \label{lem:commutativity characterisation}
    A series $f \in \series \Q \Sigma$ is commutative if, and only if,
    \begin{align}
        \label{eq:commutativity characterisation}
        \deriveleft a \deriveleft b f
            = \deriveleft b \deriveleft a f
                \quad\text{and}\quad
        \deriveleft a f = \deriveright a f,
        \qquad \forall a, b \in \Sigma.
    \end{align}
\end{lemma}
Thanks to the effective closure properties of $P$-finite series
(\cref{lem:P-finite:closure properties,cor:closure under right derivatives:special}),
all the series mentioned in~\cref{eq:commutativity characterisation} are effectively $P$-finite
when $P$ is a special product rule.
The equalities can be decided by~\cref{thm:equivalence:decidability},
obtaining the following result.
\begin{theorem}
    The commutativity problem for $P$-finite series
    is decidable for every special product rule $P$.
\end{theorem}

\subparagraph*{Complexity of the commutativity problem.}

In fact, commutativity not only reduces to equivalence,
but also equivalence reduces to commutativity, so that the two problems are effectively inter-reducible.
This relies on the following additional closure property.

Fix a finite alphabet $\Sigma$ and tuple of series $f = \tuple{f_a}_{a \in \Sigma} \in \series \Q \Sigma^\Sigma$.
We say that $g \in \series \Q \Sigma$ is a \emph{left anti-derivative} of $f$
if $\deriveleft a g = f_a$ for all $a \in \Sigma$.
Clearly, once we fix the initial condition $g_\e \in \Q$, anti-derivatives are unique.
$P$-finite series are closed under left anti-derivatives,
for every product rule $P$ (not necessarily special).
\begin{agdalemma}[General/FinitelyGenerated-AntiDerivatives][Closure under left anti-derivatives]
    \label{lem:closure under anti-derivatives}
    Consider an arbitrary product rule $P$ (not necessarily special).
    If $g \in \series \Q \Sigma$ is an anti-derivative
    of a tuple of $P$-finite series $f = \tuple{f_a}_{a \in \Sigma}$,
    then $g$ is also $P$-finite.
\end{agdalemma}
\begin{proof}
    Let $g$ be an anti-derivative of $\tuple{f_a}_{a \in \Sigma}$,
    where each $f_a$ is $P$-finite with finitely many generators
    $F_a \subseteq_{\text{fin}} \series \Q \Sigma$. 
    Then $g$ is $P$-finite with generators
    $\set g \cup \setof{f_a}{a \in \Sigma} \cup \bigcup_{a \in \Sigma} F_a$.
\end{proof}

Thanks to anti-derivatives, we can reduce equivalence to commutativity.
\begin{lemma}
    \label{lem:equivalence reduces to commutativity}
    Let $P$ be a product rule (not necessarily special).
    The equivalence problem for $P$-finite series
    effectively reduces to the commutativity problem for $P$-finite series.
\end{lemma}
\begin{proof}
    Let $f \in \series \Q \Sigma$ be an arbitrary series.
    We construct a series $g \in \series \Q \Gamma$ \st
    \begin{align*}
        \text{$g$ is commutative} 
            \quad\text{ iff }\quad
                f = \zero.
    \end{align*}
    To this end, consider two fresh alphabet symbols $a, b \notin \Sigma$
    and let $\Gamma := \Sigma \cup \set{a, b}$.
    Let $g$ be the unique series \st\-
    $g_x = 0$ for every $x \in \Gamma \cup \set \e$,
    and for every $x, y \in \Gamma$ let
    \begin{align*}
        \deriveleft {xy} g =
            \begin{cases}
                f & \text{if } x = a,  y = b, \\
                \zero & \text{otherwise}.
            \end{cases}
    \end{align*}
    Clearly if $f = \zero$, then $g = \zero$ and in particular it is commutative.
    Conversely, assume there is an input word $w \in \Sigma^*$ \st~$f_w \neq 0$.
    Then $g_{abw} = f_w$ but $g_{baw} = 0$,
    and thus $g$ is not commutative.
    
    The proof is concluded by noticing that $g$ is constructed from $f$ and $\zero$ by taking anti-derivatives.
    If $f$ is $P$-finite, then so is $g$ by~\cref{lem:closure under anti-derivatives}.
\end{proof}

When $P$ is a special product rule,
equivalence of $P$-finite series is Ackermann-hard
for the product rule $P = \dot x \dot y$ of the Hadamard product
(see~\cref{rem:complexity of equivalence}).
Therefore, it follows from~\cref{lem:equivalence reduces to commutativity}
that commutativity inherits this lower bound in this case.
Conversely, since commutativity efficiently reduces to equivalence,
commutativity decidable with Ackermann complexity for every special product rule $P$.

\end{document}